\documentclass[conference]{IEEEtran}
\usepackage{amsfonts}

\usepackage{amsmath}
\usepackage{mathabx}
\usepackage{comment}
\usepackage{cancel}

\usepackage{amsfonts}

 \usepackage{amsmath}
\usepackage{mathabx}
\usepackage{comment}
\usepackage{stmaryrd}

\ifCLASSINFOpdf
\else
\fi
\begin{document}
%
\title{
The Logic of Counting Query Answers
}


\author{
\IEEEauthorblockN{Hubie Chen}
\IEEEauthorblockA{Departamento LSI, 
Universidad del Pa\'{i}s Vasco\\
\emph{and} IKERBASQUE, Basque Foundation for Science\\
E-20018 San Sebasti\'{a}n, Spain
}
\and
\IEEEauthorblockN{Stefan Mengel}
\IEEEauthorblockA{
CNRS\\CRIL UMR 8188\\
France
}
}


%



\maketitle

\begin{abstract}
We consider the problem of counting the number of answers
to a first-order formula on a finite structure.
We present and study an extension of first-order logic
in which algorithms for this counting problem can be 
naturally and conveniently expressed, 
in senses that are made precise and that are motivated
by the wish to understand tractable cases of the counting problem.
\end{abstract}


%
\IEEEpeerreviewmaketitle

\newtheorem{theorem}{Theorem}[section]
\newtheorem{conjecture}[theorem]{Conjecture}
\newtheorem{corollary}[theorem]{Corollary}
\newtheorem{examples}[theorem]{Examples}
\newtheorem{proposition}[theorem]{Proposition}
\newtheorem{prop}[theorem]{Proposition}
\newtheorem{lemma}[theorem]{Lemma}
\newtheorem{definition}[theorem]{Definition}
\newtheorem{observation}[theorem]{Observation}
\newtheorem{remark}[theorem]{Remark}
\newtheorem{examplecore}[theorem]{Example}

\newenvironment{example}
  {\begin{examplecore}\rm}
  {\hfill $\Box$\end{examplecore}}
%




\newcommand{\sh}{\sharp}
\newcommand{\countp}{\mathsf{count}}

\newcommand{\mc}{\mathsf{MC}}

\newcommand{\nats}{\mathbb{N}}
\newcommand{\N}{\mathbb{N}}
\newcommand{\Q}{\mathbb{Q}}
\newcommand{\Z}{\mathbb{Z}}

\newcommand{\str}{\mathrm{str}}

\newcommand{\rela}{\mathbf{A}}
\newcommand{\relb}{\mathbf{B}}
\newcommand{\relc}{\mathbf{C}}
\newcommand{\reld}{\mathbf{D}}
\newcommand{\rele}{\mathbf{E}}
\newcommand{\relt}{\mathbf{T}}
\newcommand{\relp}{\mathbf{P}}
\newcommand{\relf}{\mathbf{F}}
\newcommand{\relg}{\mathbf{G}}
\newcommand{\relh}{\mathbf{H}}
\newcommand{\reli}{\mathbf{I}}
\newcommand{\relm}{\mathbf{M}}
\newcommand{\reln}{\mathbf{N}}
\newcommand{\relr}{\mathbf{R}}

\newcommand{\calP}{\mathcal{P}}
\newcommand{\calE}{\mathcal{E}}

\newcommand{\free}{\mathrm{free}}
\newcommand{\closed}{\mathrm{closed}}
\newcommand{\lib}{\mathrm{lib}}

\newcommand{\id}{\mathrm{id}}
\newcommand{\surj}{\mathrm{surj}}
\newcommand{\tw}{\mathsf{tw}}
\newcommand{\starsize}{\mathsf{starsize}}
\newcommand{\qaw}{\mathsf{qaw}}
\newcommand{\topp}{\mathsf{top}}
\newcommand{\atom}{\mathsf{atom}}

\newcommand{\res}{\upharpoonright}
\newcommand{\aug}{\mathsf{aug}}
\newcommand{\width}{\mathsf{width}}
\newcommand{\shwidth}{\sh\textup{-}\width}
\newcommand{\contract}{\mathsf{contract}}

\newcommand{\p}{\mathsf{P}}
\newcommand{\np}{\mathsf{NP}}

\newcommand{\FO}{\mathsf{FO}}
\newcommand{\PP}{\mathsf{PP}}
\newcommand{\EP}{\mathsf{EP}}

\newcommand{\fo}{\mathsf{FO}}
\newcommand{\pp}{\mathsf{PP}}
\newcommand{\ep}{\mathsf{EP}}

\newcommand{\pcountp}{\smallp\countp}

\newcommand{\dom}{\mathrm{dom}}
\newcommand{\powfin}{\wp_{\mathsf{fin}}}

\newcommand{\pn}[1]{\textsc{#1}}

\newcommand{\smallp}{\ensuremath{p\textup{-}}}
\newcommand{\clique}{\pn{Clique}}
\newcommand{\sclique}{\pn{\#Clique}}
\newcommand{\pclique}{\smallp\clique}
\newcommand{\psclique}{\smallp\sclique}

\newcommand{\param}[1]{\mathsf{param}\textup{-}#1}

\newcommand{\hnote}[1]{}
\newcommand{\snote}[1]{[Stefan note: #1]}

\newcommand{\lics}[1]{#1}
\newcommand{\pods}[1]{}
\newcommand{\licspods}[2]{\lics{#1}\pods{#2}}

\newenvironment{proof}{\noindent\textbf{Proof\/}.}{\hfill$\Box$\medskip}


\section{Introduction}

\subsection{Overview}

The computational problem of evaluating a
logical formula
on a finite relational structure is of central interest in 
database theory and logic.
In the context of database theory, this problem is
often referred to as \emph{query evaluation}, as
it models the posing of a query to a database, in a well-acknowledged way: 
the formula represents the query,
and the structure represents the database.
We refer to the results of such an evaluation as
\emph{answers}; 
logically, these are the satisfying assignments of the formula on the structure.
The particular case of this problem 
where the formula is a sentence 
is known as \emph{model checking}.
In the sequel, 
we assume by default that (unless otherwise specified)
all formulas and sentences under discussion
are first-order and relational.

This article concerns 
the problem of \emph{counting query answers}:
given a first-order formula and a finite structure,
output the \emph{number} of answers
(for previous studies, see for example the works~\cite{PichlerSkritek11-counting,GrecoScarcello14-counting,DurandMengel13-structuralcounting,ChenMengel15-pp-icdt}).
This problem is a counting version of query evaluation,
and
generalizes model checking, which can be viewed as the 
particular case thereof where one is given a sentence and structure,
and wants to decide if the number of answers is $1$ or $0$, 
corresponding to whether or not the empty assignment is satisfying.
Motivation for studying this counting problem 
stems both from
basic and fundamental interest, 
and from application scenarios:
all practical query languages
supported by database management systems have
a counting operator, and
it has indeed been argued~\cite{GrecoScarcello14-counting}
 that database queries with counting
are at the basis of decision support systems 
that handle large data volume.

With a first-order formula $\phi$ in hand, if one is interested in 
counting the \emph{number} of answers to $\phi$ on given
structures, 
\licspods{it is natural to}{it is natural---from the perspective of database theory---to} 
inquire if there is a language or logic in which 
one can directly express the mapping 
that provides, for each structure,
the number of answers to $\phi$.
Such a logic could serve as a target language
into which first-order formulas of interest
(in the mentioned sense)
could be compiled,
and then optimized, rewritten, and evaluated.
This article presents and studies such a logic, 
\emph{$\sh$-logic}, wherein the evaluation of a sentence
on a structure yields an integer value.
From the database-theoretic viewpoint, our presentation
of $\sh$-logic amounts to the introduction of
a query language designed particularly for counting answers.
We show that $\sh$-logic enjoys and balances the following properties.
\begin{itemize}

\item {\bf Expressivity.} In a sense made precise, 
$\sh$-logic allows for the expression of 
known efficient algorithms
for tractable cases of the counting query answers problem.
Moreover, this expression is (in our view) direct and clean,
and illustrates that $\sh$-logic captures precisely the
key computational primitives required by these algorithms;
this capture, in turn, justifies the particular definition
of $\sh$-logic.

\item {\bf Optimizability.}
Minimizing a crucial measure known as \emph{width} 
can be performed computably in an expressive fragment of $\sh$-logic;
this amounts to the fragment supporting 
an optimal form of
query optimization, relative to this quantity.

\end{itemize}

Our hope is that this article will contribute to and 
invigorate a broader investigation 
of query languages for counting answers.
Such an investigation could address issues such as
the identification of desirable theoretical properties of
such query languages, and techniques for performing
query rewriting, optimization, and evaluation.






\subsection{Background: complexity}

As has been previously articulated in the literature,
a typical situation in the database setting is the evaluation of a relatively short formula on a relatively large structure.  Consequently, it has been argued that, in measuring the time complexity of query evaluation tasks, one could reasonably allow a slow (non-polynomial-time) preprocessing of the formula, so long as the desired evaluation can be performed in polynomial time following the preprocessing~\cite{PapadimitriouYannakakis99-database,FlumGrohe06-parameterizedcomplexity}.
Relaxing polynomial-time computation to allow arbitrary preprocessing of a \emph{parameter} of a problem instance yields, in essence, the notion of \emph{fixed-parameter tractability}.
This notion is at the core of 
and is the primary tractability notion in 
\emph{parameterized complexity theory}, which provides a taxonomy for classifying problems where each instance has an associated parameter.  
Following this motivation, 
whenever the problem of counting query answers
(or the model checking problem) is considered using
parameterized complexity, in this article,
the formula is taken to be the parameter.

The problem of counting query answers
is well-known to be computationally intractable.
It is possible to restrict this problem by
considering restricted classes of queries,
and then trying to understand which classes of queries 
are computationally well-behaved in that they give
rise to a tractable case of the general problem.
Precisely, for a class $\Phi$ of first-order formulas,
define $\countp(\Phi)$ to be the problem
where an instance is a formula $\phi \in \Phi$ paired
with a finite structure $\relb$,
and the output is the number of answers of $\phi$ on $\relb$.
We hence have a family of problems, one problem
$\countp(\Phi)$
for each such formula class $\Phi$,
and one can inquire which of these problems are
tractable (and which are not).
We will also have cause to consider the case of 
model checking, so,
when $\Phi$ is a class of first-order sentences,
define $\mc(\Phi)$ to be the \emph{model checking}
problem where an instance is a sentence $\phi \in \Phi$
paired with a finite structure $\relb$,
and the output is yes or no depending on whether or not
$\relb$ satisfies $\phi$.
The complexity of the problem family $\mc(\Phi)$ has
been considered in numerous papers, 
such as~\cite{Grohe07-otherside,AdlerWeyer09-treewidth,ChenDalmau12-decomposingquantified,ChenMarx13-blocksorted,Chen14-existentialpositive,ChenMueller14-hierarchy,Chen14-frontier}.

\subsection{Background: width}

\emph{Width} is a syntactic measure of logical formulas.
The \emph{width} of a first-order formula 
$\phi$ is defined
as the maximum number of free variables over all subformulas 
of $\phi$.
In studies of the problem family $\mc(\Phi)$,
width has emerged as a crucially relevant measure; 
we now explain how.

Say that a class $\Phi$ of first-order formulas has
\emph{bounded width} if there exists a constant $k \geq 1$
such that each $\phi \in \Phi$ has width at most $k$.
It is now well-known that bounded width sentence classes are
computationally desirable for model checking, made precise as follows.

\begin{observation} \cite{Immerman82-fo-expressibility,Vardi95-boundedvariable}
\label{obs:mc-ptime}
Suppose that $\Phi$ is a class of sentences having bounded width.
The computational problem $\mc(\Phi)$
is polynomial-time decidable, via the algorithm that (given a sentence and a structure) simply computes the set of satisfying assignments for each subformula, inductively.
\end{observation}

As suggested above, 
it is known that a problem is fixed-parameter tractable 
if, after performing preprocessing on the parameter,
an instance 
can be resolved in polynomial time~\cite[Theorem 1.37]{FlumGrohe06-parameterizedcomplexity}.
For a problem $\mc(\Phi)$, one can readily observe
that if each sentence $\phi \in \Phi$ can be algorithmically
translated to a sentence lying in a class having bounded width,
then $\mc(\Phi)$ is fixed-parameter tractable as a consequence
of Observation~\ref{obs:mc-ptime}.
This can be formalized as follows.

\begin{observation} \cite{Chen14-existentialpositive}
\label{obs:mc-fpt}
Suppose that $\Phi$ is a sentence class.
The following condition, which we will refer to as the
\emph{classical condition}, is sufficient for
the problem $\mc(\Phi)$ to be fixed-parameter tractable:
there exists a sentence class $\Phi'$ having bounded width
and an algorithm $f$ that computes, for each $\phi \in \Phi$,
a logically equivalent sentence $f(\phi)$ that is in $\Phi'$.
\end{observation}

Research on the problem family $\mc(\Phi)$ has succeeded in 
obtaining classifications on classes $\Phi$ of bounded arity
where
the quantifiers and connectives are restricted
(see for examples~\cite{Grohe07-otherside,ChenMarx13-blocksorted,Chen14-existentialpositive}).
An example relevant to the present article 
is a study~\cite{Chen14-existentialpositive}
 of existential positive logic
(by which we mean the positive fragment of first-order logic
consisting of formulas built from 
atoms, $\wedge$, $\vee$, and $\exists$);
making crucial use of a hardness result by 
Grohe~\cite{Grohe07-otherside}, 
this work 
observed that when $\Phi$ is a class of existential positive sentences
having bounded arity,
the problem $\mc(\Phi)$ is fixed-parameter tractable
if and only if the classical condition applies to $\Phi$.
(Here and elsewhere in our discussion, we assume the standard
complexity-theoretic hypothesis that FPT $\neq$ W[1].)
That is, the sufficient condition for fixed-parameter tractability
identified by Observation~\ref{obs:mc-fpt} is
the \emph{exclusive} explanation for fixed-parameter tractability,
in the described setting of existential positive logic.
Let us remark that the query preprocessing algorithm $f$ here 
(in Observation~\ref{obs:mc-fpt})
is related to and akin to the database notion of a query optimizer that computes a query execution plan.

Existential positive logic is a natural fragment of 
first-order logic~\cite{Rossman08-hompreservation} and is studied heavily
in database theory.
Existential positive formulas include and are 
semantically equivalent to so-called
\emph{unions of conjunctive queries}, also known as
\emph{select-project-join-union queries},
which have been argued to be the most 
common database queries~\cite{AbiteboulHullVianu95-foundationsdatabases}.
Recently, the present authors 
generalized the mentioned dichotomy 
on existential positive sentence classes,
by presenting
a classification theorem~\cite{ChenMengel16-pods-ep-counting}
describing the fixed-parameter tractable problems
of the form $\countp(\Phi)$, where $\Phi$ is a bounded arity class
of existential positive formulas.  The classification theorem is
in fact a trichotomy theorem, which demonstrates that
the studied problems $\countp(\Phi)$ can exhibit three types of complexity behavior.

\subsection{Contributions}

As a means of introducing our contributions,
we here wish to highlight a 
conceptual point: 
the applicability of the above \emph{classical condition}
indicates that for the model checking problem,
\emph{first-order logic itself can be used as a model of computation 
in which desirable, efficient algorithms can be expressed.}
This condition posits the existence of an algorithm that translates
first-order sentences to a  polynomial-time evaluable format;
and, the particular format used therein is 
that of a first-order sentence!
Let us highlight that here, logic can be viewed as playing 
two complementary roles: on the one hand, the 
computational problems of interest are phrased directly in terms of logic;
on the other hand,  appropriate algorithmic solutions to this problem
are themselves describable by logical sentences.

Inspired by this perspective of 
\emph{logic as a useful model of computation},
the present work was motivated by the desire to develop a 
logic that could serve as a useful model of computation for the
problem of counting query answers---analogously to 
how first-order logic itself
serves as a useful model of computation for the model checking problem.

Let us point out some desiderata that such a logic ought to fulfill.
First, recall that in the problem of counting query answers,
an instance is a first-order formula paired with a structure, and the
output is the \emph{number} of answers.  
The hope, then, 
would be to be able to translate a first-order formula $\phi$
to a sentence $\psi$ in the logic such that evaluating $\psi$
on a structure returns the number of answers to $\phi$ on the structure;
hence, in the logic, the evaluation of a sentence on a structure
ought to return a \emph{numerical quantity}, instead of
a propositional value as in usual first-order logic.
Second, the logic must accommodate the fact that, while 
the classical condition explains all tractable cases of
model checking in existential positive logic (in the sense made precise),
there are classes $\Phi$ of formulas 
for which the problem $\countp(\Phi)$
is fixed-parameter tractable, but
on which the classical condition
does not hold---in the sense that $\Phi$ does not have
bounded width, even if each formula therein
may be replaced with a logically equivalent one.\footnote{
An example of such a class is
$\Theta = \{ \theta_n ~|~ n \geq 1 \}$
where 
$\theta_n(x_1, \ldots, x_n) = U_1(x_1) \wedge \cdots \wedge U_n(x_n)$.
The number of answers for $\theta_n$ on any structure $\relb$
is straightforwardly verified to be $\prod_{i=1}^n |U_i^{\relb}|$.
Given $\theta_n$ and $\relb$, this quantity can indeed be computed
in polynomial time;
however, it can be verified that the classical condition
does not apply to $\Theta$, in the mentioned sense.
}

In this article, we introduce and study a logic,
which we call \emph{$\sh$-logic} 
and which possesses the sought-after characteristics just described.
In $\sh$-logic, 
the evaluation of a \emph{$\sh$-sentence}
(a type of formula in $\sh$-logic) on a structure
returns an integer value.

\licspods{
   \subsubsection{\bf $\sh$-logic, a preview}
}{
	\subsubsection{$\sh$-logic: a preview}
}
\label{subsubsect:sh-logic-preview}

Syntactically, $\sh$-logic consists of \emph{$\sh$-formulas};
each $\sh$-formula $\phi$ 
has an associated set of free variables,
denoted by $\free(\phi)$,
Let $\relb$ be a structure, let $\phi$ be a $\sh$-formula
over the signature of $\relb$,
and let $h: \free(\phi) \to B$ be an assignment.
Semantically,
evaluating $\phi$ with respect to $\relb$ and $h$
returns an integer value, as opposed to a propositional value
(as for a fo-formula).
To present and discuss the semantics
of $\sh$-logic,
we will notationally use $[\relb, \phi]$
to denote the mapping that takes an assignment
$h: \free(\phi) \to B$ to the 
corresponding
integer value
(that is, the integer value provided to the triple consisting of
$\phi$, $\relb$, and $h$).


To offer the reader a feel for the syntax and semantics
of $\sh$-logic, we provide a discussion of 
some example formulas.

\begin{example}
Define 
$\phi(x,y,z)$ to be the formula $E(x,y) \wedge F(x,z)$.
The first type of $\sh$-formula is a \emph{casting}
of a fo-formula; define $\psi$ to be
$C(E(x,y), \{ x, y \})$, which is a $\sh$-formula
with free variables $\{ x, y \}$.
Let $\relb$ be a structure, and let
$h: \{ x, y \} \to B$ be an assignment;
$[\relb, \psi](h)$ is equal to $1$ or $0$
depending on whether or not $\relb, h \models E(x,y)$ (respectively).
Once one has a $\sh$-formula, it is possible to define
a further $\sh$-formula by \emph{projecting} free variables.
For example, $P \{ y \} \psi$ is a $\sh$-formula
with free variables $\{ x \}$.
When $g: \{ x \} \to B$ is an assignment,
the value $[\relb, P \{ y \} \psi](g)$ is the
number of extensions $h: \{ x, y \} \to B$ of $g$ such that
$\relb, h \models E(x,y)$.
In an analogous fashion, one may define
$\psi'$ to be the $\sh$-formula
$C(F(x,z), \{ x, z \})$;
then,
when $g: \{ x \} \to B$ is an assignment,
the value $[\relb, P \{ z \} \psi'](g)$ is the
number of extensions $h': \{ x, z \} \to B$ of $g$
such that 
$\relb, h' \models F(x, z)$.

Now, let $g: \{ x \} \to B$ be an assignment.
Observe that the number of extensions 
$h^+: \{ x, y, z \} \to B$ 
of $g$ such that $\relb, h^+ \models E(x,y) \wedge F(x,z)$
is equal to the product of $[\relb, P \{ y \} \psi](g)$
and $[\relb, P \{ z \}  \psi'](g)$.
A product connective $\times$ is provided by $\sh$-logic,
and said product is equal to
$[\relb, (P \{ y \} \psi) \times (P \{ z \} \psi')](g)$.
Finally, $\theta = P \{ x \} ((P \{ y \} \psi) \times (P \{ z \} \psi'))$
is a $\sh$-formula with no free variables;
letting $\emptyset$ denote the empty assignment,
$[\relb, \theta](\emptyset)$ will be equal to 
the sum, over all assignments $g: \{ x \} \to B$,
of $[\relb, (P \{ y \} \psi) \times (P \{ z \} \psi')](g)$,
which is equal to the number of assignments $f: \{ x, y, z \} \to B$
such that $\relb, f \models \phi$.
Phrased in terminology that will be defined precisely, the $\sh$-formula
$\theta$ \emph{represents} the fo-formula $\phi(x,y,z)$.
\end{example}

\licspods{
	\subsubsection{\bf $\sh$-logic, features}
}{
  \subsubsection{$\sh$-logic: features}
}

From the discussed viewpoint of bounded width as
an explanation for the tractability of model checking,
the relationship of the counting query answers problem
to $\sh$-logic is strongly analogous to
the relationship of model checking
to usual first-order logic.  The following parallel of
Observation~\ref{obs:mc-ptime} holds.
(Note that the \emph{width} of a formula in $\sh$-logic
will be defined in a natural way.)

\begin{observation}
\label{obs:count-ptime}
Suppose that $\Psi$ is a class of $\sh$-sentences
having bounded width.
Then, the computational problem of evaluating
a sentence $\psi \in \Psi$ on a finite structure
is polynomial-time computable.
(See Proposition~\ref{prop:bd-width-evaluation-in-sh-logic}
for a precise statement and further information.)
\end{observation}

As expressed, one purpose of $\sh$-logic is to allow for the translation of
a first-order formula $\phi$ to a $\sh$-sentence $\psi$
such that $\psi$ \emph{represents} $\phi$ in that, 
for any structure $\relb$,
the number of answers to $\phi$ on $\relb$
is equal to the quantity that results from evaluating $\psi$ on $\relb$.
The following is an immediate consequence of 
the previous observation, and a parallel of
Observation~\ref{obs:mc-fpt}.

\begin{observation}
\label{obs:count-fpt}
Suppose that $\Phi$ is a first-order formula class.
The following condition, which we will refer to as the \emph{counting condition},
is sufficient for the problem $\countp(\Phi)$ to be 
fixed-parameter tractable:
there exists a $\sh$-sentence class $\Psi$ having bounded width
and an algorithm $f$ that computes,
for each $\phi \in \Phi$, a $\sh$-sentence $f(\phi)$ 
that is in $\Psi$ and
that
represents~$\phi$.
\end{observation}

Of course, Observation~\ref{obs:count-fpt} is only of interest
if the \emph{counting condition}
possesses \emph{explanatory power},
that is, only if this condition
allows one to explain the fixed-parameter tractability of
problems $\countp(\Phi)$ having interest.
We in fact show that,
in the context of existential positive queries,
the \emph{counting condition} has maximal explanatory power:
\begin{theorem}
Let $\Phi$ be any class of existential positive queries having
bounded arity.
If
$\countp(\Phi)$
is tractable,
the counting condition applies to $\Phi$.
 (See Theorem~\ref{thm:passing-to-rep-under-tractability-condition} for a precise statement.)
\end{theorem}
That is,  the
counting condition is the 
\emph{exclusive} explanation for the tractability of
$\countp(\Phi)$ in this existential positive setting,
providing an analog to the result that
the classical condition is the
exclusive explanation for the tractability of
$\mc(\Phi)$ in the existential positive setting.
On a conceptual level, we view this result as strong evidence that,
for the problem of counting query answers, $\sh$-logic
is a useful, expressive model of computation in which relevant,
efficient algorithms can be presented.
This result is obtained as an immediate consequence of two theorems:
\begin{itemize}

\item We show that when such a problem 
$\countp(\Phi)$ is 
tractable, then there exists a bounded width class $\Psi$ 
of $\sh$-sentences such that each $\phi \in \Phi$ 
has a representation in $\Psi$
(Theorem~\ref{thm:shepequi}).

\item We prove that there is a \emph{minimization algorithm} that,
given an existential positive formula, computes a representation
of minimum width
(Theorem~\ref{thm:shepcompute}).

\end{itemize}

The latter theorem, which we view as a key contribution
in and of itself, can be read as demonstrating that
$\sh$-logic is well-characterized and well-understood
as a model of computation: conceiving of a $\sh$-sentence
representation of an existential positive formula
as a computational procedure for counting query answers,
this theorem provides a minimization algorithm that always outputs
an \emph{optimal} procedure for a given existential positive formula,
where optimality here is measured using~width.

In short, our presentation and study of $\sh$-logic
forwards the discussed use of logic as a means
for expressing computationally desirable procedures;
in particular, $\sh$-logic allows for the direct expression
of procedures for counting query answers.

\licspods{
\subsubsection{\bf Counting homomorphisms}
}{
	\subsubsection{Counting homomorphisms}
}

The problem of counting the number of homomorphisms
from a given \emph{source} structure $\rela$ 
to a given \emph{target} structure $\relb$ 
arises and has been studied in numerous contexts~\cite{GroheThurley11-counting-and-partition,Bulatov13-counting}.
This problem can be viewed as the special case of counting query
answers where the formula is a quantifier-free conjunction of atoms;
there is indeed a correspondence that allows one to pass
from a source structure $\rela$ to such a formula $\phi_{\rela}$
which
originates from the classical work~\cite{ChandraMerlin77-optimal}
(and which is explained in Section~\ref{sect:preliminaries}).
This problem is now well-known to be polynomial-time tractable
under a constant treewidth bound on the permitted
source structures; indeed, the corresponding algorithm,
which performs dynamic programming over a tree decomposition of $\rela$,
has received 
a textbook treatment~\cite[Section 5.3]{FominKratsch10-exact-exponential-algs}.

We discuss how, 
from a tree decomposition for a structure $\rela$,
one can compute, in polynomial time, a $\sh$-sentence $\psi$
that has
width at most the width of the given decomposition (plus one),
and
that represents $\rela$ in the sense that 
evaluating $\psi$ on an arbitrary structure $\relb$
always yields the number of homomorphisms from $\rela$ to $\relb$
(Example~\ref{ex:tree-decomposition-to-sh-formula}).
Combining this result with Observation~\ref{obs:count-ptime},
we obtain that the algorithm of this well-known polynomial-time
tractability result
can be expressed in $\sh$-logic (Proposition~\ref{prop:solving-tw-bounded-homom-counting}).
Indeed, we believe that the resulting $\sh$-sentences
accurately, faithfully, and cleanly describe the execution
of this algorithm.


\licspods{
\subsubsection{\bf A dual perspective on a classical theorem of Lov\'{a}sz}
}{
\subsubsection{A dual perspective on a classical theorem of Lov\'{a}sz}
}
Fix a relational signature $\tau$; in the following discussion,
all structures are finite and on $\tau$.
Let $\relb$ be a structure,
let $\str[\tau]$ denote the class of finite structures on $\tau$,
and let $L(\relb)$ be the vector from $\Q^{\str[\tau]}$
that maps a structure $\rela \in \str[\tau]$ 
to the number of homomorphisms
from $\rela$ to $\relb$.  A classical theorem of Lov\'{a}sz
states that, for any two structures $\relb$, $\relb'$,
it holds that $L(\relb) = L(\relb')$ iff
$\relb$ and $\relb'$ are isomorphic~\cite{Lovasz67-operations-structures}.

When one is concerned with homomorphisms from 
one structure $\rela$ to another structure $\relb$,
sometimes, the structure $\rela$ is referred to as the 
\emph{left-hand} structure and the structure $\relb$
is referred to as the \emph{right-hand} structure.
The vectors studied by Lov\'{a}sz indicate,
for a structure $\relb$, the number of homomorphisms
coming \emph{from} each possible left-hand side structure.
One can naturally formulate a dual vector, as follows.
For any structure $\rela$, define $R(\rela)$
to be the vector from $\Q^{\str[\tau]}$
that maps a structure $\relb \in \str[\tau]$
to the number of homomorphisms from $\rela$ to $\relb$.
That is, the vector $R(\rela)$ indicates,
for a structure $\rela$,
the number of homomorphisms \emph{to}
each possible right-hand side structure.

Our previous work~\cite[Theorem 5.4]{ChenMengel16-pods-ep-counting}
implied
a dual of Lov\'{a}sz's theorem, namely, that 
for any two structures $\rela$, $\rela'$,
it holds that $R(\rela) = R(\rela')$ iff
$\rela$ and $\rela'$ are isomorphic.
In the present work, we prove and use a 
natural generalization of this fact,
namely, that for any finite sequence 
$\rela_1, \ldots, \rela_k$ of pairwise non-isomorphic
structures, the vectors $R(\rela_i)$ are 
linearly independent (Theorem~\ref{thm:linearcombination}).
(We view these vectors as over the rational numbers $\Q$
so as to be able to properly discuss linearly independence.)
This linear independence theorem is used as a key tool
to establish the correctness of our minimization algorithm.
We believe that it should play an important role in future
studies of counting query answers, and that
the techniques and concepts that its proof requires
 may be of independent and future interest.
 Indeed, to prove this theorem,
 we extend notions and techniques from the work
 of Lov\'{a}sz~\cite{Lovasz67-operations-structures}; 
 for example, we introduce and crucially use a notion of multivariate 
 polynomial associated to a primitive positive formula.


\subsection{Discussion}

Logics with counting mechanisms have been considered
in finite model theory and descriptive complexity;
one well-known example is the counting logic studied
by Immerman and Lander~\cite{ImmermanLander90-describing-graphs}.
A typical motivation in these areas for studying such logics
is the desire to extend first-order logic 
in order to capture properties
not expressible in first-order logic.
This motivation contrasts somewhat with ours here;
our objective is to introduce logics that allow for
the relatively direct expression of useful algorithms
for the problem of counting query answers.
We believe that it could be of interest to try to understand
the relationship (if any) between existing counting logics
and the logics studied in the present work.

We wish to emphasize that, as regards our present motivations,
our logic\footnote{Precisely, 
the \emph{$\sh\ep$-formulas}, the
fragment of our logic that we focus on.}
trades off expressivity and computability properties 
in an extremely desirable fashion.
On the one hand, the algorithms
for the tractable cases of $\countp(\Phi)$
(where $\Phi$ is existential positive)
can be expressed in our logic, as described above;
on the other hand, our \emph{minimization algorithm}
described above evidences that the measure of \emph{width}
can be computably minimized in our logic, and is 
thus in a certain sense well-characterized.

Previous work established that
there is no algorithm for minimizing width
in positive first-order logic~\cite[Section 5]{BovaChen14-width-ep}.
As a consequence, there is no algorithm for minimizing
width in any logic 
that (1) includes positive first-order logic as a fragment
and (2) where a width minimization algorithm would imply 
a width minimization algorithm for positive first-order logic.
Thus, such a logic
would provably not exhibit the identified 
expressivity-computability
tradeoff that our logic enjoys.
To the best of our knowledge, our width minimization algorithm
is the first such algorithm for a logic with a 
form of counting mechanism; again, we view this as one contribution of this article.

We believe that our introduction of $\sh$-logic 
may open up further research directions.  
One particular question for future research that we may pose
is whether or not there are 
Ehrenfeucht-Fra\"{i}ss\'{e} style games for proving
inexpressibility in bounded width fragments of our logic.


\section{Preliminaries}
\label{sect:preliminaries}

\subsection{Logic}
We assume basic familiarity with the syntax and semantics of first-order logic.
In this article, we focus on relational first-order logic
where equality is not built-in to the logic.
Hence, each \emph{vocabulary/signature} under discussion consists
only of relation symbols.
We assume structures under discussion to be \emph{finite}
(that is, have finite universe);
nonetheless, 
we sometimes describe structures as \emph{finite} for emphasis.
We use the letters $\rela$, $\relb$, $\ldots$ to denote structures,
and the corresponding letters $A$, $B$, $\ldots$
to denote their respective universes.
When $\rela, \relb$ are structures over the same signature $\tau$,
a \emph{homomorphism} from $\rela$ to $\relb$
is a mapping $h: A \to B$ such that, for each $R \in \tau$
and each tuple $(a_1, \ldots, a_k) \in R^{\rela}$,
it holds that $(h(a_1), \ldots, h(a_k)) \in R^{\relb}$.

We use the term \emph{fo-formula} to refer to a first-order formula.
An \emph{ep-formula} (short for \emph{existential positive formula}) 
is a fo-formula built from
\emph{atoms} 
(by which we refer to predicate applications of the form $R(v_1, \ldots, v_k)$, where $R$ is a relation symbol and the $v_i$ are variables),
conjunction ($\wedge$), disjunction ($\vee$),
and existential quantification ($\exists$).
A \emph{pp-formula}
(short for \emph{primitive positive formula}) 
is an ep-formula where disjunction does not occur.
An fo-formula is \emph{prenex} if it has the form
$Q_1 v_1 \ldots Q_n v_n \theta$ where $\theta$ is quantifier-free,
that is, if all quantifiers occur in the front of the formula.
The set of free variables of a formula $\phi$ is denoted by
$\free(\phi)$ and is defined as usual;
a formula $\phi$ is a \emph{sentence} if $\free(\phi) = \emptyset$.
We define an \emph{ep-sentence} to be an ep-formula that is
a sentence, and define \emph{fo-sentence} and \emph{pp-sentence} similarly.

We now present some definitions and conventions that
are not totally standard.
A primary concern in this article 
is in counting satisfying assignments of fo-formulas
on a finite structure.
The count is sensitive to the set of variables over which
assignments are considered; and, we will 
\emph{sometimes} want to count
relative to a set of variables that is strictly larger than
the set of free variables.
Hence, we will often associate with each fo-formula $\phi$
a set $V$ of variables called the \emph{liberal variables},
denoted by $\lib(\phi)$, for which it is required that
$\lib(\phi) \supseteq \free(\phi)$.
We generally assume that the variables in 
$\lib(\phi) \setminus \free(\phi)$ are not used in $\phi$.
To indicate that $V$ is the set of liberal variables of $\phi$,
we often use the notation $\phi(V)$;
we also use $\phi(v_1, \ldots, v_n)$, where the $v_i$ 
are a listing of the liberal variables.
Relative to a formula $\phi(V)$, when~$\relb$ is a structure,
we will use $\phi(\relb)$ to denote the set of
assignments $f: V \to B$ such that $\relb, f \models \phi$.
We call an fo-formula $\phi$ \emph{free} if $\free(\phi) \neq \emptyset$,
and
\emph{liberal} if 
$\lib(\phi)$ is defined and $\lib(\phi) \neq \emptyset$.

\begin{example}
Consider the formula 
$$\phi(x,y,z) = E(x,y) \vee F(y,z).$$
Define $\psi(x,y,z) = E(x,y)$ and $\psi'(x,y,z) = F(y,z)$.
The notation is intended to indicate that
$$\lib(\phi) = \lib(\psi) = \lib(\psi') = \{ x, y, z \}.$$
In the context of studying $\phi$, it is natural to define
$\lib(\psi)$ and $\lib(\psi')$ to be $\{ x, y, z \}$;
under these definitions, 
it holds that $\phi(\relb) = \psi(\relb) \cup \psi'(\relb)$,
but in general this would not hold in the case that
(say) $\lib(\psi)$ was defined as $\{ x, y \}$
(which set is equal to $\free(\psi)$).
\end{example}

\subsection{pp-formulas}
It is well-known~\cite{ChandraMerlin77-optimal}
that there is a correspondence between prenex pp-formulas
and relational structures.
In particular, 
each prenex pp-formula $\phi(S)$ (on signature $\tau$) 
with $\lib(\phi) = S$
may be viewed as a
pair $(\rela, S)$ 
consisting of a structure $\rela$ (on $\tau$) 
and a set $S$;
the universe $A$ of $\rela$ is
the union of $S$ with the variables appearing in $\phi$,
and the following condition 
defines the relations of $\rela$: for each $R \in \tau$,
a tuple $(a_1, \ldots, a_k) \in A^k$
is in $R^{\rela}$ 
if and only if $R(a_1, \ldots, a_k)$ appears in $\phi$.
In the other direction, 
such a pair $(\rela, S)$ can be viewed as a prenex pp-formula $\phi(S)$
where all variables in $A \setminus S$ are quantified
and the atoms of $\phi$ are defined according to the above condition.
A basic known fact~\cite{ChandraMerlin77-optimal}
that we will use
is that
when $\phi(S)$ is a pp-formula corresponding to the pair $(\rela, S)$,
$\relb$ is an arbitrary structure, and $f: S \to B$ is an arbitrary map,
it holds that 
$\relb, f \models \phi(S)$ if and only if
there is an extension $f'$ of $f$ that is a homomorphism
from~$\rela$ to $\relb$.
\emph{We will freely interchange between the
 structure view
and the usual notion
of a prenex pp-formula.}
For a prenex pp-formula specified as a pair $(\rela,S)$, 
we typically assume that $S \subseteq A$.

\begin{example}
\label{ex:pp-views}
Consider the pp-formula that is given as
$\phi(u,v,w,x) = \exists y( E(u,v) \wedge F(w,y))$.
To convert $\phi$ to a structure $\rela$,
we take the universe $A$ of $\rela$ to be the union of
$\lib(\phi)$ with the variables appearing in $\phi$,
so $A = \{ u, v, w, x, y \}$.
The relations of $\rela$ are as defined above,
so $E^{\rela} = \{ (u,v) \}$ and $F^{\rela} = \{ (w,y) \}$.
The resulting pair representation of $\phi$
is $(\rela, \{ u, v, w, x \})$.
\end{example}

\subsection{Graphs}
Throughout the paper, all graphs under discussion 
should be assumed to be undirected by default.

To every prenex pp-formula $(\rela, S)$
we assign a graph 
whose vertex set is $A \cup S$ 
and where two vertices are connected by an edge if they appear together in a tuple of a relation of $\rela$.
A prenex pp-formula $(\rela, S)$ is called \emph{connected} if its graph is connected. 
A prenex pp-formula $(\rela', S')$ is a
\emph{component} of a prenex pp-formula $(\rela, S)$
over the same signature $\tau$
if there exists a set $C$ that forms a connected component
of the graph of $(\rela, S)$, where
$S' = S \cap C$, and
for each relation $R \in \tau$, 
a tuple $(a_1, \ldots, a_k)$ is in $R^{\rela'}$
if and only if 
$(a_1, \ldots, a_k) \in R^{\rela} \cap C^k$.



Note that when this holds, the graph of $(\rela', S')$
is the connected component of the graph of $(\rela, S)$ on vertices $C$.
We will use the fact that, if $\phi(V)$ is a prenex pp-formula
and $\phi_1(V_1), \ldots, \phi_k(V_k)$ is a list of its components,
then for any finite structure $\relb$,
it holds that 
$|\phi(\relb)| = \prod_{i=1}^k |\phi_i(\relb)|$.

\begin{example}
Consider the formula $\phi$ from
Example~\ref{ex:pp-views}.
The connected components of the graph of $\phi$ 
can be readily verified to be
$\{ u, v \}$,  $\{ w, y \}$, and
$\{ x \}$.
Hence, the pp-formula $\phi$ has 3 components,
which can be readily verified to be
$\phi_1(u,v) = E(u,v)$, 
$\phi_2(w) = \exists y F(w,y)$, and
$\phi_3(x) = \top$;
here, $\top$ denotes the empty conjunction.
\end{example}

\subsection{Treewidth}
We give some basic facts about tree decompositions
and treewidth; see for example~\cite{FlumGrohe06-parameterizedcomplexity} for more details.

A tree decomposition of a graph $G=(V(G),E(G))$ is a pair 
$(T, (B_t)_{t\in V(T)})$ where $T$ is a tree and $(B_t)_{t\in V(T)}$ is a family of subsets of $V(G)$ such that (1)
for every $v\in V(G)$, the set $\{t\in V(T)\mid v\in B_t\}$ is non-empty and connected in $T$, and
(2)
for every edge $uv\in E(G)$, there is a $t\in V(T)$ such that $u,v\in B_t$.
We also denote $B_t$ using the notation $B(t)$.
The \emph{width} of a tree decomposition $(T, (B_t)_{t\in V(T)})$ is 
defined as 
$\max\{|B_t|: t\in V(T)\}-1$. 
The treewidth $\tw(G)$ of $G$ is the minimum width over all the tree decompositions of $G$. 
Computing tree decompositions of minimal width is fixed-parameter tractable parameterized by the 
treewidth~\cite{Bodlaender96}. 

\newcommand{\nspace}{\vspace{4pt}}

A tree decomposition 
is called \emph{nice} if 
its tree $T$ is rooted and
every $t\in V(T)$ is of one of the following types:
\begin{itemize}
 \item (leaf) $t$ has no children and $|B(t)|=1$.

\nspace

 \item (introduce) $t$ has one child $t'$ and $B(t)= B(t') \cup \{v\}$ for a vertex $v\in V\setminus B(t')$.
 
\nspace

 \item (forget) $t$ has one child $t'$ and $B(t)= B(t') \setminus \{v\}$ for a vertex $v\in B(t')$.
 
\nspace

 \item (join) $t$ has two children $t_1, t_2$ with $B(t) = 
 B(t_1)= B(t_2)$.

\end{itemize}
It is well-known that a 
width $k$ tree decomposition  of $G$ can be converted to a 
width $k$ nice tree decomposition, in polynomial time.





%

\section{$\sh$-logic}

In this section, we present the syntax and semantics of $\sh$-logic,
as well as some associated terminology.
Syntactically, $\sh$-logic consists of \emph{$\sh$-formulas};
each $\sh$-formula $\phi$ 
has an associated set of free variables,
denoted by $\free(\phi)$,
as well as an associated set of \emph{closed} variables,
denoted by $\closed(\phi)$.
(At this point, the reader may wish to recall the
preview of $\sh$-logic presented in 
Section~\ref{subsubsect:sh-logic-preview}.)

\subsection{Syntax}

We define $\sh$-formulas inductively, as follows.

\newcommand{\itemspace}{\vspace{6pt}}
\begin{itemize}

\item 
\licspods{
$C(\phi, L)$ is a $\sh$-formula 
if $\phi$ is a fo-formula \\and $L \supseteq \free(\phi)$.	
}{
$C(\phi, L)$ is a $\sh$-formula 
if $\phi$ is a fo-formula \\and $L \supseteq \free(\phi)$.
}

\licspods{
Define
 $\free(C(\phi,L)) = L$ and $\closed(C(\phi,L)) = \emptyset$.
}{
Define
 $\free(C(\phi,L)) = L$\\ and $\closed(C(\phi,L)) = \emptyset$.
}

\itemspace

\item 
$PV \phi$ is a $\sh$-formula
if $\phi$ is a $\sh$-formula and $V$ is
a set of variables with 
$V \cap \closed(\phi) = \emptyset$.

Define $\free(PV \phi) = \free(\phi) \setminus V$\\ 
and $\closed(PV \phi) = V \cup \closed(\phi)$.

\itemspace

\item 
$EV \phi$ is a $\sh$-formula
if $\phi$ is a $\sh$-formula
and $V$ is a set of variables
with $V \cap (\free(\phi) \cup \closed(\phi)) = \emptyset$.

Define $\free(EV \phi) = V \cup \free(\phi)$\\
and $\closed(EV \phi) = \closed(\phi)$.

\itemspace

\item 
$\phi \times \phi'$ is a $\sh$-formula
if $\phi$ and $\phi'$ are $\sh$-formulas
with $\free(\phi) = \free(\phi')$
and $\closed(\phi) \cap \closed(\phi') = \emptyset$.

Define $\free(\phi \times \phi') = \free(\phi)$\\
and $\closed(\phi \times \phi') = \closed(\phi) \cup \closed(\phi')$.

\itemspace

\item 
$\phi + \phi'$ is a $\sh$-formula
if $\phi$ and $\phi'$ are $\sh$-formulas
with $\free(\phi) = \free(\phi')$.

Define $\free(\phi + \phi') = \free(\phi)$\\
and $\closed(\phi + \phi') = \closed(\phi) \cup \closed(\phi')$.

\itemspace

\item $n$ is a $\sh$-formula if $n \in \Z$.

\licspods{
Define $\free(n) = \emptyset$
and $\closed(n) = \emptyset$.	
}{
Define $\free(n) = \emptyset$\\
and $\closed(n) = \emptyset$.
}
\end{itemize}

A formula $C(\phi,L)$ can be thought of as the casting
of a fo-formula $\phi$ into a $\sh$-formula;
the $P$ quantifier can be thought of as projecting or
closing variables; and the $E$ quantifier can be thought of
as expanding the set of free variables.
The connectives $\times$ and $+$ perform the usual arithmetic
operations.
We remark that, for each $\sh$-formula $\phi$,
it holds that $\free(\phi)$ and $\closed(\phi)$ are disjoint;
this is straightforwardly verified by induction.

Let $\psi$ be a $\sh$-formula.
We say that a $\sh$-formula $\theta$ is
a \emph{$\sh$-subformula} of $\psi$ 
if $\theta$ is used in the inductive formation of $\psi$.
We say that a fo-formula $\theta$ is
a \emph{fo-subformula} of $\psi$ 
if $\psi$ contains a subformula $C(\phi,L)$
where $\theta$ is a subformula of $\phi$.

A \emph{subformula} of $\psi$ is a $\sh$-subformula or fo-subformula
of $\psi$.
We define $\width(\psi)$ to be the maximum of $|\free(\theta)|$
over all subformulas $\theta$ of $\psi$,
and $\shwidth(\psi)$ to be the maximum of $|\free(\theta)|$
over all $\sh$-subformulas $\theta$ of $\psi$.
We say that $\psi$ is a \emph{$\sh$-sentence} if $\free(\psi) = \emptyset$.

We define a \emph{$\sh\PP$-formula} to be a $\sh$-formula
where, in each $\sh$-subformula of the form $C(\phi,L)$,
$\phi$ is a pp-formula; the notion of 
\emph{$\sh\EP$-formula} is defined analogously.
We define a \emph{$\sh\PP$-sentence} to be a $\sh\PP$-formula
that is a $\sh$-sentence, and we define
a \emph{$\sh\EP$-sentence} similarly.

\subsection{Semantics}

We define the semantics of our logic.
For each structure $\relb$, each $\sh$-formula $\psi$ on the 
vocabulary of $\relb$, and 
each assignment $h: \free(\psi) \to B$,
we define $[\relb,\psi](h)$ recursively, as follows.



\begin{itemize}

\item When $C(\phi,L)$ is a $\sh$-formula,\\
$[\relb,C(\phi,L)](h) = 1$ if $\relb, h \models \phi$;\\
$[\relb,C(\phi,L)](h) = 0$ otherwise.


\itemspace

\item When $PV \phi$ is a $\sh$-formula,\\
$[\relb,PV \phi](h) = \sum_{h'} [\relb,\phi](h')$,\\
where the sum is over all extensions \\
$h': \free(\phi) \cup V \to B$
of $h$.

\itemspace

\item When $EV \phi$ is a $\sh$-formula,\\
$[\relb, EV \phi](h) = [\relb, \phi](h \res \free(\phi))$.

\itemspace

\item When $\phi \times \phi'$ is a $\sh$-formula,\\
$[\relb, \phi \times \phi'](h) = [\relb,\phi](h) \cdot [\relb,\phi'](h)$.

\itemspace

\item When $\phi + \phi'$ is a $\sh$-formula,\\
$[\relb, \phi + \phi'](h) = [\relb,\phi](h) + [\relb,\phi'](h)$.

\itemspace

\item 
\licspods{
When $n$ is a $\sh$-formula,
$[\relb, n](h) = n$.	
}{
When $n$ is a $\sh$-formula,\\
$[\relb, n](h) = n$.
}



\end{itemize}
We consider two $\sh$-formulas $\phi$, $\phi'$ with
$\free(\phi) = \free(\phi')$ to be \emph{logically equivalent} if
for each structure $\relb$, it holds that
$[\relb, \phi] = [\relb, \phi']$.
A $\sh$-sentence $\psi$ \emph{represents} or is a \emph{representation} of
a fo-formula $\phi(V)$ if
for each finite structure $\relb$,
it holds that $|\phi(\relb)| = [\relb, \psi](\emptyset)$,
where $\emptyset$ is the empty assignment.
For simplicity, when $\psi$ is a $\sh$-sentence,
we will typically write
$[\relb,\psi]$ in place of
$[\relb, \psi](\emptyset)$.
We will use the term \emph{$\sh\pp$-representation} to refer to
a representation that is a $\sh\pp$-formula,
and define \emph{$\sh\ep$-representation} similarly.

We make the basic observation that each fo-formula
has a representation.

\begin{prop}
\label{prop:simple-representation}
For each fo-formula $\phi(V)$, 
the $\sh$-sentence $PV C(\phi,V)$
is a representation of $\phi(V)$.
\end{prop}

We also observe that, when a constant width bound is assumed,
evaluation of $\sh$-sentences can be performed in polynomial time.

\begin{prop}
\label{prop:bd-width-evaluation-in-sh-logic}
For each $k \geq 1$, 
there exists a polynomial-time algorithm that,
given a finite structure $\relb$ and 
a $\sh$-sentence $\psi$ having $\width(\psi) \leq k$,
computes $[\relb, \psi]$.
\end{prop}

\begin{proof}
For each subformula $\theta$ of $\psi$ and each
mapping $h: \free(\theta) \to B$,
the algorithm computes $[\relb,\theta](h)$ in the case 
that $\theta$ is a $\sh$-subformula,
and determines whether or not $\relb, h \models \theta$
in the case that $\theta$ is a fo-subformula.
This computation is performed inductively over the subformulas
of $\psi$, and in the case that $\theta$ is a $\sh$-subformula,
the just-given semantics are used.
\end{proof}

\subsection{Examples}

\begin{example}
Let us define 
$\phi(x_0,x_1,x_2,y_0,y_1,y_2)$
as the formula $\phi_0 \wedge \phi_1 \wedge \phi_2$,
where 
$\phi_i = \exists z_i T_i(x_i, x_{i+1}, y_i, z_i)$
for each $i \in \{ 0,1,2 \}$;
here, the quantity $i+1$ appearing in $x_{i+1}$ is computed modulo $3$.
These formulas are over the vocabulary $\{ T_0, T_1, T_2 \}$
having three relation symbols, each of arity $4$.

Define $\psi_i = P \{ y_i \} C(\phi_i, \{ x_0, x_1, x_2, y_i \})$
for each $i \in \{ 0,1,2 \}$.
Observe that when $\relb$ is a structure and $h: \{ x_0, x_1, x_2 \} \to B$
is a map,
$[\relb, \psi_i](h)$ gives the number of extensions
$h': \{x_0,x_1,x_2,y_i\} \to B$ of $h$ satisfying $\phi_i$ on $\relb$.
We have $\free(\psi_0) = \free(\psi_1) = \free(\psi_2) = \{x_0,x_1,x_2\}$
and $\closed(\psi_i) = \{y_i\}$.
Now consider $\psi = (\psi_0 \times \psi_1) \times \psi_2$.
It can be verified that,
for a structure $\relb$ and a map $h: \{ x_0, x_1, x_2 \} \to B$,
$[\relb, \psi](h)$ gives the number of extensions
$h': \{ x_0, x_1, x_2, y_0, y_1, y_2 \} \to B$ of $h$
satisfying $\phi$ on $\relb$.
It follows that the $\sh$-sentence $\theta = P\{x_0,x_1,x_2\} \psi$
is a representation of $\phi$.
Since 
$$3 = |\free(\psi)| 
= |\free(\psi_0)| 
= |\free(\psi_1)| 
= |\free(\psi_2)|,$$
we obtain that
the representation $\theta$
has width equal to 
$$\max(3,\width(\psi_0),\width(\psi_1),\width(\psi_2))=4.$$

As a further remark, consider, as an example, the subformula
$C(\phi_0,\{x_0,x_1,x_2,y_0\})$ of $\psi_0$.
It holds that $x_2 \notin \free(\phi_0)$,
and so
$E\{x_2\} C(\phi_0,\{x_0,x_1,y_0\})$ 
is a $\sh$-formula and is logically equivalent to
$C(\phi_0,\{x_0,x_1,x_2,y_0\})$.
\end{example}


\begin{example}
\label{ex:tree-decomposition-to-sh-formula}
Consider a prenex pp-formula $\phi$ without quantifiers
whose structure view has the form $(\rela, A)$.
Suppose that $(T, (B(t)))$ is a width $k$ tree decomposition
of the graph of $(\rela, A)$;
without loss of generality, we may assume that this tree decomposition
is nice, and that the root node $r$ of $T$ has $B(r) = \emptyset$.
We explain how to give a representation of $\phi$
having width $\leq k + 1$.

Say that a node $u$ of $T$ is \emph{below} a node $t$ of $T$
if $t$ occurs on the unique simple path from $u$ to the root of $T$
(this is understood to hold in particular when $u = t$).
When $t$ is a node of $T$, define $B(\leq t)$
to be the union of $B(u)$ over all nodes $u$ that are below $t$.

\newcommand{\tspace}{\vspace{6pt}}

We show that, for each node $t$ of $T$,
there exists a $\sh$-formula $\psi_t$ such that:
\begin{itemize}
\item $\free(\psi_t) = B(t)$,

\tspace

\item $\closed(\psi_t) = B(\leq t) \setminus B(t)$, and

\tspace

\item for any structure $\reld$,
the value $[\reld, \psi_t](h)$ is equal to the number
of extensions $h': B(\leq t) \to D$ of $h$ such that
$(\reld, h')$ satisfies 
each atom $R(a_1, \ldots, a_k)$ of $\phi$ 
whose variables
all fall into a bag $B(u)$, with $u$ below $t$.

\tspace

\end{itemize}
We give a $\sh$-formula depending on the type of the node $t$;
we use the notation from the definition of nice tree decomposition.
\begin{itemize}

\item (introduce)\\
$\psi_t = (Ev \psi_{t'}) \times C(\alpha_1, B(t)) \times \cdots \times C(\alpha_m, B(t))$

where the $\alpha_i$ are the atoms of $\phi$
whose variables fall into $B(t)$.
Note that the presence of the $C(\alpha_i, B(t))$
ensures that $[\reld, \psi_t](h) = 0$ if $(\reld, h)$ does not satisfy
all of the $\alpha_i$.

\tspace

\item (leaf) $\psi_t$ is defined as in the previous case,
except $(Ev \psi_{t'})$ is omitted from the product.

\tspace

\item (forget) $\psi_t = P \{ v \} \psi_{t'}$.

\tspace

\item (join) $\psi_t = \psi_{t_1} \times \psi_{t_2}$.
In this case,
$\psi_{t_1} \times \psi_{t_2}$ 
is a $\sh$-formula since
$\free(\psi_{t_1}) = B(t_1) = B(t_2) = \free(\psi_{t_2})$,
and, by the definition of tree decomposition,
it holds that the sets
$B(\leq t_1) \setminus B(t_1) = \closed(\psi_{t_1})$ and
$B(\leq t_2) \setminus B(t_2) = \closed(\psi_{t_2})$ are disjoint.

\tspace

\end{itemize}
The desired representation is $\psi_r$.
The claim on the width of $\psi_r$ holds, for we have the following:
each
$\sh$-subformula $\psi'$ of $\psi_r$
has $\free(\psi') = B(w)$ for a node $w$ of $T$,
and each fo-subformula of $\psi_r$ is 
either $\top$ or an atom
whose variables fall into a bag $B(t)$.
\end{example}

From the discussion in Example~\ref{ex:tree-decomposition-to-sh-formula},
we obtain the following.

\begin{prop}
\label{prop:solving-tw-bounded-homom-counting}
Let $k \geq 1$.
Consider the problem of computing, given a pair $(\rela,\relb)$
of relational structures (over the same signature)
where $\rela$ has treewidth $\leq k$,
the number of homomorphisms from $\rela$ to $\relb$.
This problem can be solved by the polynomial-time algorithm
that computes a nice, width $k$ tree decomposition
(using a known polynomial-time algorithm for this task);
computes, from $\rela$, the representation $\psi_r$
given by Example~\ref{ex:tree-decomposition-to-sh-formula};
and, invokes the algorithm of
Proposition~\ref{prop:bd-width-evaluation-in-sh-logic}
to compute $[\relb,\psi_r]$.
\end{prop}

\section{Main theorems}


\subsection{Statements}

The treewidth of a prenex pp-formula $\phi$, 
denoted here by $\tw(\phi)$, is defined as the treewidth of the graph of $\phi$.
The following notions are adapted from~\cite{ChenMengel15-pp-icdt,ChenMengel16-pods-ep-counting}.
Let $(\rela, S)$ be a prenex pp-formula with graph $G$.
An \emph{$\exists$-component} of $(\rela, S)$
is a graph of the form $G[W']$ where there exists
$W \subseteq A$ that is a connected component of $G[A \setminus S]$
and $W'$ is the union of $W$ with all vertices in $S$ having
an edge to $W$.  Define $\contract(\rela, S)$ to be the
graph on vertices $S$ obtained by starting from $G[S]$
and adding an edge between any two vertices that appear together
in an $\exists$-component of $(\rela, S)$.\footnote{
  Note that in previous articles~\cite{ChenMengel15-pp-icdt,ChenMengel16-pods-ep-counting}
  this graph was defined in terms of the 
  core of the given formula.
}
\begin{example}
Let $\phi(u_1, u_2, u_3, u_4)$ be the pp-formula
\begin{center}
$\exists x_1 \exists x_2 \exists x_3
(R(u_1, x_1, x_2) \wedge
S(u_2, x_2) \wedge
T(u_2, x_3) \wedge 
U(u_3, x_3) \wedge
P(u_3, u_4))$.
\end{center}
The graph $G$ of $\phi$
has vertex set 
\begin{center}
$V = \{ u_1, u_2, u_3, u_4, x_1, x_2, x_3 \}$
\end{center}
and edge set
\begin{center}
$E = \{ 
\{ u_1, x_1 \}, \{ u_1, x_2 \}, \{ x_1, x_2 \},
\{ u_2, x_2 \}, $

$ \{ u_2, x_3 \}, \{ u_3, x_3 \},
\{ u_3, u_4 \}
\}$.
\end{center}
Note that, if we were to view $\phi$ as a pair
$(\rela, S)$, the structure $\rela$ would have universe $V$
and we would have $S = \lib(\phi) = \{ u_1, u_2, u_3, u_4 \}$.
There are two connected components of $G[A \setminus S]$,
namely,
$W_1 = \{ x_1, x_2 \}$ and $W_2 = \{ x_3 \}$.
(Note that, in contrast, $G$ itself is connected.)
The $\exists$-components of $\phi$ are thus
$W'_1 = W_1 \cup \{ u_1, u_2 \}$ and
$W'_2 = W_2 \cup \{ u_2, u_3 \}$.
The graph $\contract(\phi)$ is the graph on vertices 
$S = \{ u_1, u_2, u_3, u_4 \}$ obtained by starting
from $G[S]$, which has the single edge $\{ u_3, u_4 \}$,
and adding the edges $\{ u_1, u_2 \}$ and $\{ u_2, u_3 \}$.
\end{example}
A \emph{core} of $(\rela, S)$ is a prenex pp-formula $(c(\rela),S)$
where $c$ is an $\rela$-endomorphism 
fixing each $s \in S$ 
that has minimum image size.
By $c(\rela)$, we mean the structure with universe $c(A)$
and where $R^{c(\rela)} = \{ c(t) ~|~ t \in R^{\rela} \}$;
here, $c(t)$ denotes the tuple 
obtained by applying $c$ to each entry of $t$.
It is known that any core of a prenex pp-formula $\phi$
is logically equivalent to $\phi$~\cite{ChandraMerlin77-optimal}.
We speak of \emph{the core} of a pp-formula, as it
is unique up to isomorphism; this follows from the
basic theory of cores~\cite{HellNesetril92-core}.

A class $\Phi$ of prenex pp-formulas satisfies
the \emph{contraction condition} 
if the class containing each graph $\contract(\psi)$,
where $\psi$ is the core of a formula in $\Phi$,
has bounded treewidth;
$\Phi$ satisfies the
\emph{tractability condition} if it satisfies the contraction
condition and the cores of formulas in $\Phi$ have bounded treewidth.
Previous work showed that,
for bounded arity $\Phi$,
 $\countp(\Phi)$ is fixed-parameter
tractable if $\Phi$ satisfies the tractability condition;
interreducible with the parameterized clique problem
if $\Phi$ satisfies the contraction condition but not the tractability condition; and as hard as the parameterized counting clique problem
otherwise (see~\cite{ChenMengel15-pp-icdt} for a precise statement).


We first study representations of pp-formulas, obtaining
the following theorems.

\begin{theorem}\label{thm:shppequi}
Let $\Phi$ be a class of prenex pp-formulas.
\begin{itemize}

\item 
The class $\Phi$ satisfies the tractability condition
if and only if 
there exists $k \geq 1$ such that
each formula in $\Phi$ has a 
$\sh \pp$-representation $\phi'$  such that $\width(\phi') \leq k$.

\item
The class $\Phi$ satisfies the contraction condition
if and only if
there exists $k \geq 1$ such that
each formula in $\Phi$ has a
$\sh\pp$-representation $\phi'$ such that $\shwidth(\phi') \leq k$.
\end{itemize}

\end{theorem}

\begin{theorem}\label{thm:shppcompute}
There exists an algorithm that,
given a prenex pp-formula $\phi$, 
outputs a $\sh\pp$-representation $\psi$ of $\phi$ of minimum width.
\end{theorem}

Building on this understanding of pp-formulas,
we are then able to achieve general versions of these theorems
for ep-formulas.
Previous work showed that for any class $\Phi$ of ep-formulas,
there exists a class $\Phi^+$ of prenex pp-formulas such that
the problems $\countp(\Phi)$ and $\countp(\Phi^+)$ 
are interreducible
(see~\cite{ChenMengel16-pods-ep-counting}
for a precise statement),
and hence (for example)
whether or not the tractability condition holds on
$\Phi^+$ determines whether or not $\countp(\Phi)$
is fixed-parameter tractable.
For the purposes of this extended abstract,
it is sufficient to know that
the class $\Phi^+$ is essentially defined from $\Phi$ 
in the following way: for each $\phi \in \Phi$,
it is shown that the function $|\phi(\cdot)|$,
which maps a finite structure $\relb$ to $|\phi(\relb)|$,
can be written as a polynomial 
(over the integers)
in unknowns 
$|\phi_1(\cdot)|, \ldots, |\phi_m(\cdot)|$
where the $\phi_i$ are pp-formulas; the class $\Phi^+$
is defined to contain all such formulas $\phi_i$ arising in this way.
(In general, the mentioned polynomial makes use of negative integers;
this is a reason why our definition of $\sh$-logic allows
arbitrary integers, and not just natural numbers.)
For more information, 
we refer the reader to~\cite{ChenMengel16-pods-ep-counting}.

\begin{theorem}\label{thm:shepequi}
Let $\Phi$ be a class of ep-formulas.
\begin{itemize}

\item 
The class $\Phi^+$ satisfies the tractability condition
if and only if 
there exists $k \geq 1$ such that
each formula in $\Phi$ has a 
$\sh \ep$-representation $\phi'$  having $\width(\phi') \leq k$.

\item
The class $\Phi^+$ satisfies the contraction condition
if and only if
there exists $k \geq 1$ such that
each formula in $\Phi$ has a
$\sh\ep$-representation $\phi'$ having $\shwidth(\phi') \leq k$.
\end{itemize}

\end{theorem}

\begin{theorem}\label{thm:shepcompute}
There exists an algorithm that,
given an ep-formula $\phi$, 
outputs a $\sh\ep$-representation $\psi$ of $\phi$
of minimum width.
\end{theorem}

The following is a consequence of the previous two theorems.

\begin{theorem}
\label{thm:passing-to-rep-under-tractability-condition}
Let $\Phi$ be a class of ep-formulas.
If $\Phi^+$ satisfies the tractability condition,
then there exists $k \geq 1$
and an algorithm that, given a formula $\phi \in \Phi$,
computes a $\sh\ep$-representation $\phi'$ of $\phi$
having $\width(\phi') \leq k$.
\end{theorem}

\begin{proof}
The algorithm is that provided by Theorem~\ref{thm:shepcompute}.
The claim on the width follows immediately from the first part of
Theorem~\ref{thm:shepequi}.
\end{proof}

\subsection{Overviews of the proofs}

In order to prove the main theorems just presented, 
we will develop several tools spanning 
Sections~\ref{sect:quantifier-aware-width} to~\ref{sect:independence-theorem} before finally proving the main results in Section~\ref{sect:proofs-of-main-theorems}. 
Here, we offer the reader guidance by
introducing the pieces of the puzzle and 
explaining how they fit together.

As formulated above, the tractability condition 
is defined in part using the contraction condition,
and thus it appears to consist of two independent parts.
In a first step, in Section~\ref{sect:quantifier-aware-width}, 
we simplify the situation by introducing a new notion that we call \emph{quantifier-aware width} of a pp-formula. It is defined to be the minimal width of certain restricted tree decompositions of a pp-formula. We show that having bounded quantifier-aware width is equivalent to the tractability condition (Lemma~\ref{lem:qawcompare});
as a consequence of this lemma,
after Section~\ref{sect:quantifier-aware-width},  
we will not have to deal with the tractability condition directly anymore and can work with the conceptually cleaner notion of quantifier-aware width. We then go on to show that quantifier-aware width can be computed in a fixed-parameter fashion (Lemma~\ref{lem:computeqaw}) which is an important building block for our minimization algorithm. Finally, we prove that quantifier-aware width of a pp-formula is essentially equal to the width of an optimal representation by a $\sh\PP$-formula of a particularly simple type, which we call \emph{basic} (Lemma~\ref{lem:shppandpp}).
Together with the rest of the results of Section~\ref{sect:quantifier-aware-width}, this gives an important connection 
between pp-formulas and $\sh\PP$-formulas and thus bridges the gap between ordinary first-order logic and $\sh$-logic.

Section~\ref{app:normalization} gives insights into the structure of $\sh\EP$-formulas by showing that we can always turn one into 
an equivalent weighted sum of basic $\sh\PP$-formulas. This allows 
us to leverage most of the results of Section~\ref{sect:quantifier-aware-width} to general $\sh\EP$-formulas.

The perhaps most subtle but very important contribution to the proofs of the main theorem comes from the results of 
Section~\ref{sect:independence-theorem}. 
In this section, we consider sums of the form 
$\sum_{i=1}^m c_i |\phi_i(V_i)|$ where each $c_i$ is a non-zero rational constant and the $\phi_i$ are 
pp-formulas 
which are pairwise not counting equivalent.
We call such terms \emph{linear combinations};
each naturally maps a finite structure $\relb$ to the value
$\sum_{i=1}^m c_i |\phi_i(\relb)|$.
The main result of Section~\ref{sect:independence-theorem} is 
the \emph{independence theorem}
(Theorem~\ref{thm:linearcombination}) which states that for every linear combination there is a structure on which the linear combination evaluates to a non-zero value.

The independence theorem
 has concrete applications throughout the proofs of the main results in Section~\ref{sect:proofs-of-main-theorems}. The reasoning is roughly as follows: we assign to a formula two different linear combinations having two different desirable properties (for example, small width and small number of summands) but computing the same value on every finite structure. Applying Theorem~\ref{thm:linearcombination} on the difference of these linear combinations, we obtain that in fact both linear combinations are the same 
 (up to \emph{counting equivalence} of the $\phi_i$,
 a notion of equivalence to be defined). 
 Consequently, both linear combinations have the same properties and thus in particular have both of the two desirable properties. This then allows to reason about the properties of the original formula we started with.

The proofs of the main results in 
Section~\ref{sect:proofs-of-main-theorems} 
use the above tools in a rather black-box fashion. Thus the reader is invited to first skip the proofs in Sections \ref{sect:quantifier-aware-width} to~\ref{sect:independence-theorem} and 
see how everything fits together in 
Section~\ref{sect:proofs-of-main-theorems} before 
reading the proofs of the individual pieces.

\section{Quantifier-aware width}
\label{sect:quantifier-aware-width}

In this section, we introduce a new width measure of pp-formulas which we call \emph{quantifier-aware width}
 and show that it is related to the width of $\sh\PP$-formulas.

We here assume all tree decompositions of pp-formulas to be nice. 
So let $(T, (B(t))_{t\in T})$ be a nice tree decomposition of a pp-formula $\phi$. For every variable $x$ of $\phi$ let $\topp(x)$ be the vertex $t$ of $T$ that is highest in $T$ such that $x\in B(t)$. 
We call a tree decomposition of $\phi$ \emph{quantifier-aware} if for every $\exists$-component $C$ of $\phi$ and for all $x\in V(C)\setminus \free(\phi)$ and all $y\in V(C)\cap \free(\phi)$, we have that $\topp(y)$ is on the path from $\topp(x)$ to the root of $T$. We call the \emph{quantifier-aware width} of a pp-formula $\phi$, denoted by $\qaw(\phi)$, the minimal treewidth of a quantifier-aware tree decomposition of $\phi$ plus~$1$.

\lics{
\begin{remark}\label{rem:star}
The quantifier-aware width can be arbitrarily higher than their treewidth. To see this consider the formula $\phi = \exists z \bigwedge_{i\in [n]} E(x_i, z)$. The primal graph of $\phi$ is a star, so it has treewidth $1$. We claim that the quantifier-aware width of $\phi$ is $n+1$. To see this, observe first that the free variables $x_1, \ldots, x_n$ must appear above $\topp(z)$ in a bag of any quantifier-aware tree decomposition. But since $x_i z$ is an edge in the primal graph for every $i$, the variable $x_i$ must also appear in a common bag with $z$ and consequently also in $\topp(z)$. Thus $\topp(z)$ must contain $n+1$ variables. 
\end{remark}
}


We now demonstrate properties of quantifier-aware width
that will be used to establish our main theorems.

We first show that $\qaw(\phi)$ is, in a sense, characterized by 
$\tw(\phi)$ and $\tw(\contract(\phi))$, for every pp-formula $\phi$. Consequently, quantifier-aware treewidth will allow us to characterize tractable classes of pp-formulas for counting. 

\begin{lemma}\label{lem:qawcompare}
 For every pp-formula $\phi$ we have \begin{align*} &\max\{\tw(\phi), \tw(\contract(\phi))\} + 1 \\\le & \qaw(\phi) \le \tw(\phi) + \tw(\contract(\phi))+1.\end{align*}
\end{lemma}

We now give a lemma that shows how to compute the quantifier-aware treewidth, which will allow us to compute $\sh\PP$-formulas of optimal width.
\begin{lemma}\label{lem:computeqaw}
 The computation of quantifier-aware tree decompositions 
 of pp-formulas 
having minimal width 
 is fixed-parameter tractable, 
 when the parameter is taken as the quantifier-aware width.
\end{lemma}

We call a $\sh\PP$-formula $\phi$ \emph{basic}
if it does not contain $+$ nor subformulas of the form $n$, 
where $n \in \Z$. 
The following lemma demonstrates that basic $\sh\PP$-formulas 
correspond very closely to pp-formulas.

\begin{lemma}\label{lem:shppandpp}
\begin{enumerate}

 \item[a)] 
There exists an algorithm that, given
a basic $\sh\PP$-sentence $\phi'$,
 computes a pp-formula $\phi$ that $\phi'$ represents,
 such that 
$\width(\phi')\ge \qaw(\phi)$, 
and\\
$\shwidth(\phi') \ge \tw(\contract(\phi))+1$.

 \item[b)]  
 There exists an algorithm that, given a pp-formula $\phi$,
computes a basic $\sh\PP$-sentence $\phi'$ that represents $\phi$,
 such that
  $\width(\phi')\le \qaw(\phi)$, and\\
  $\shwidth(\phi') \le \tw(\contract(\phi))+1$.
\end{enumerate}
\end{lemma}

\section{Normalizing $\sh\ep$-formulas}\label{app:normalization}

We call a $\sh\pp$-formula \emph{constant} if it is only constructed from constants in $\mathbb{Z}$, $\times$, and $P$- and $E$-quantifiers.
We call a $\sh\pp$-formula
\emph{flat} if 
it is of the form $\sum_{i\in [\ell]} \psi_i\times \phi_i$ 
where the $\psi_1, \ldots, \psi_n$ are constant and $\phi_1,\ldots , \phi_\ell$ are basic $\sh\PP$-formulas.

The main result of this section is the following normalization lemma.

\begin{lemma}\label{lem:flattening}
There exists an algorithm that computes,
for a given $\sh\EP$-formula $\phi$,
a logically equivalent flat $\sh\PP$-formula $\phi'$ such that $\width(\phi')\le \width(\phi)$.
\end{lemma}


We prove a sequence of lemmas to aid us.

\begin{lemma}\label{lem:turnPP}
There exists an algorithm that computes,
for a given $\sh\EP$-formula $\phi$ of the form $C(\psi,L)$, 
a logically equivalent $\sh\PP$-formula $\phi'$ such that $\width(\phi')\le \width(\phi)$.
\end{lemma}

We call an $\sh$-formula \emph{$+$-free} if it does not contain  $+$.

\begin{lemma}\label{lem:normaddition}
  There exists an algorithm that computes, for a given $\sh\EP$-formula $\phi$, a logically equivalent $\sh\PP$-formula $\phi'$ of the form $\sum_{i=1}^s \phi_i$ where the $\phi_i$ are $+$-free such that $\width(\phi')\le \width(\phi)$.
\end{lemma}
\begin{proof} 
 The proof is by straightforward induction on the structure of $\phi$, pushing all occurences of $+$ up in the formula;
 for instance, one proves that $PV (\psi_1 + \psi_2)$
 is logically equivalent to $PV \psi_1 + PV \psi_2$.
 The base case $\phi=C(\psi,L)$ is Lemma~\ref{lem:turnPP}. 
\end{proof}

The proof of the following lemma is by a straightforward induction.

\begin{lemma}\label{lem:normconst}
 There exists an algorithm that computes, for a given constant $\sh\PP$-formula $\phi$, a logically equivalent $\sh\PP$-formula $\phi'=EV_1 PV_2 n$ with $n\in \mathbb{Z}$ such that $\width(\phi')\le \width(\phi)$.
\end{lemma}

\begin{lemma}\label{lem:normmult}
  There exists an algorithm that computes, for a given $+$-free $\sh\PP$-formula $\phi$, a logically equivalent $\sh\PP$-formula $\phi'=\psi_1\times \psi_2$ where $\psi_1$ is constant and $\psi_2$ is basic such that $\width(\phi')\le \width(\phi)$.
\end{lemma}
\begin{proof} The proof is again straightforward induction in the style of Lemma~\ref{lem:normaddition}. We consider only the the case of $P$-quantifiers which is the only case that is not completely clear from the definition.

So let $\phi = PV \phi'$ where $\phi'=\phi_1\times \phi_2$ such that $\phi_1$ is constant and $\phi_2$ is basic. Note that by Lemma~\ref{lem:normconst} we may assume that $\phi_1= PV_1 EV_2 n$ for some $n\in \mathbb{Z}$. We claim that $\phi$ is logically equivalent to $\phi_1' \times PV \phi_2$ where $\phi_1' = PV_1 E(V_2\setminus V)$. 
To see this, consider a structure $\relb$ 
and an assignment $h$ to $\phi$. Then
$[\relb, PV \phi'](h) = \sum_{h'} [\relb, \phi_1\times \phi_2](h')
 = \sum_{h'} ([\relb, \phi_1](h')\cdot [\relb, \phi_2](h'))$
where the $h'$ are as in the definition. Now for an arbitrary assignment $h''$ to $\phi'$, $[\relb, \phi_1](h'')= [\relb, \phi_1](h')= [\relb, \phi_1'](h)$ for all $h'$. 
\licspods{Consequently, we have
\begin{align*}
 [\relb, PV \phi'](h) &= [\relb, \phi_1](h'')\cdot\sum_{h'}  [\relb, \phi_2](h')\\
 &= [\relb, \phi_1'](h)\cdot\sum_{h'}  [\relb, \phi_2](h')\\
 &= [\relb, \phi_1'](h)\cdot[\relb, PV\phi_2](h)\\
 &= [\relb, \phi_1'\times PV\phi_2](h).
 \end{align*}
}{
Consequently, we have
$[\relb, PV \phi'](h) 
= [\relb, \phi_1](h'')\cdot\sum_{h'}  [\relb, \phi_2](h')\\
 = [\relb, \phi_1'](h)\cdot\sum_{h'}  [\relb, \phi_2](h')
 = [\relb, \phi_1'](h)\cdot[\relb, PV\phi_2](h)\\
 = [\relb, \phi_1'\times PV\phi_2](h).$
}
\end{proof}


\licspods{
\begin{proof} (of Lemma~\ref{lem:flattening})
}{
\begin{proof}[of Lemma~\ref{lem:flattening}]
}
 First use Lemma~\ref{lem:normaddition} to turn $\phi$ into a sum of $+$-free $\sh\PP$-formulas. After this, apply 
 Lemma~\ref{lem:normmult}
to each of the summands.
\end{proof}

\section{Independence theorem}
\label{sect:independence-theorem}


In this section, we establish a key tool for reasoning
 about $\sh\EP$-formulas.

Define two fo-formulas 
$\phi(V)$, $\phi'(V')$
to be \emph{counting equivalent} 
if, for each structure $\relb$, it holds that
$|\phi(\relb)| = |\phi'(\relb)|$;
note that a decidable characterization
of counting equivalence on pp-formulas
is known (\cite[Theorem 5.4]{ChenMengel16-pods-ep-counting}).

In the scope of this article, define a \emph{linear combination}
to be an expression of the form
$\sum_{i=1}^m c_i |\phi_i(V_i)|$,
where $m \geq 1$, 
each $c_i$ is a non-zero rational number, 
and the $\phi_i(V_i)$
are pp-formulas that are pairwise not counting equivalent.
Here, the notation $|\phi_i(V_i)|$ 
is intended to indicate the function 
that maps each finite structure $\reld$
to the value $|\phi_i(\reld)|$.
So, each linear combination $\ell$ naturally induces a mapping
$\ell(\cdot)$
from finite structures to $\Q$,
namely, the map given by
$\ell(\reld) = \sum_{i=1}^m c_i |\phi_i(\reld)|$.
The following theorem will be key for our understanding of
equivalence of $\sh\EP$-formulas.

\begin{theorem} (Independence theorem)
\label{thm:linearcombination}
For any linear combination $\ell$,
there exists a finite structure $\reld$ such that
$\ell(\reld) \neq 0$.
\end{theorem}

We devote the rest of this section to proving this theorem.
By multiplying all values $c_i$ by a multiple of their denominators,
we can and will assume that each value $c_i$ is an integer.

In order to establish this theorem, 
we will make use of
the following lemma.

\begin{lemma}
\label{lemma:commutativity} 
(Lov\'{a}sz \cite{Lovasz67-operations-structures})
For
each univariate polynomial $p$ with positive integer coefficients
and each finite structure $\relb$, 
there exists a finite structure $p(\relb)$
such that,
for each connected liberal pp-formula $\phi$,
it holds that
$|\phi(p(\relb))| = p(|\phi(\relb)|)$.
\end{lemma}

The following lemma shows that
certain pp-formulas can be controlled independently of each other.

\begin{lemma} 
\label{lemma:orthogonality}
Let $\phi_1(S_1), \ldots, \phi_n(S_n)$ be connected liberal pp-formulas
(over signature $\tau$)
that are pairwise not counting equivalent.
Then for every $m \geq 2$, there exist structures
$(\relb_{(a_1, \ldots, a_n)} ~|~ (a_1, \ldots, a_n) \in [m]^n)$
and injective functions
$f_1, \ldots, f_n: [m] \to \N$ 
such that for each $(a_1, \ldots, a_n) \in [m]^n$ and 
each $i \in [n]$,
it holds that $|\phi_i(\relb_{(a_1, \ldots, a_n)})| = f_i(a_i)$.

Moreover, when $\rela$ is any structure on which
$|\phi_i(\rela)| > 0$ for each $i \in [n]$, 
each structure
$\relb_{(a_1, \ldots, a_n)}$ can be chosen 
in the form
$\rela \times \rele$,
where $\rele$ is a structure such that 
$|\phi(\rele)| > 0$ for all pp-formulas $\phi$ over $\tau$.
\end{lemma}

\begin{proof} 
As each $\phi_i(S_i)$ is connected and liberal,
the result 
\cite[Theorem 5.14]{ChenMengel16-pods-ep-counting}
ensures
that there exists a structure $\relc'$ such that the values
$|\phi_i(\relc')|$ are pairwise different,
and where $|\phi(\relc')| > 0$ for all pp-formulas $\phi$.
By taking a sufficiently large power $P$ of $\relc'$, 
we may obtain that 
for the structure $\relc = \relc'^P \times \rela$,
the values
$c_i = |\phi_i(\relc)|$ are pairwise different.
For each $(a_1, \ldots, a_n) \in [m]^n$,
define $p_{(a_1, \ldots, a_n)}$ to be a univariate polynomial 
over the rationals that
evaluates to $0$ at $0$, and to $a_i$ at $c_i$ (for each $i \in [n]$).
Define $D$ to be the absolute value of the product of all denominators
of coefficients in the defined polynomials.
Set 
$p'_{(a_1, \ldots, a_n)} = D \cdot 
p_{(a_1, \ldots, a_n)}$; each such polynomial has 
integer coefficients.
Next, set
$p^{-}_{(a_1, \ldots, a_n)}$ 
to be the restriction of
$p'_{(a_1, \ldots, a_n)}$ to summands with negative coefficients.
Define
$p''_{(a_1, \ldots, a_n)}$ 
to be
$p'_{(a_1, \ldots, a_n)} + 
2 \sum_{(a_1, \ldots, a_n) \in [m]^n} 
(-p^{-}_{(a_1, \ldots, a_n)})$.
Now, for each $(a_1, \ldots, a_n) \in [m]^n$,
define the structure $\relb_{(a_1, \ldots, a_n)}$
as $p''_{(a_1, \ldots, a_n)}(\relc)$;
for each $i \in [n]$,
we have
$$|\phi_i(\relb_{(a_1, \ldots, a_n)})| = 
 |\phi_i(p''_{(a_1, \ldots,a _n)}(\relc))| =
 p''_{(a_1, \ldots, a_n)}(c_i);$$ 
 the second equality here holds
 by the Lemma~\ref{lemma:commutativity}.  
 From these equalities and
 the definitions of 
 $p_{(a_1, \ldots, a_n)}$
 and
 $p''_{(a_1, \ldots, a_n)}$,
 it is straightforward to verify that the defined structures
 have the desired property.
Our claim concerning each $B_{(a_1, \ldots, a_n)}$
having the form $\rela \times \rele$ holds,
 as $p''_{(a_1, \ldots, a_n)}(0) = 0$ 
 (for each $(a_1, \ldots, a_n) \in [m]^n$),
 implying that the structures $\relb_{(a_1, \ldots, a_n)}$ provided
 can be obtained in the form $\relc \times \cdot$, which
 has the form $\rela \times \cdot$.
\end{proof}


\newcommand{\llb}{\llbracket}
\newcommand{\rrb}{\rrbracket}

We now introduce a highly useful notion, that of
\emph{component polynomial}.
Fix a set $V$ of liberal variables.
Denote by $\calE$ the set of counting equivalence classes of
liberal connected pp-formulas (with liberal variables from $V$).
A \emph{component polynomial} $q$ is a multivariate polynomial
with integer coefficients over variables
$\{ X_e ~|~ e \in \calE \}$.
For  any finite structure $\relb$,
we define the value of $q$ evaluated on $\relb$,
denoted by $q \llb \relb \rrb$,
as the integer value obtained by evaluating $q$
when each $X_e$ is given the value $|\phi_e(\relb)|$,
for a formula $\phi_e \in e$.
The following is our main theorem on component polynomials.

\begin{theorem}
\label{thm:component-polynomials}
When $q$ is a non-zero component polynomial,
there exists a finite structure $\relb$ such that
$q\llb \relb \rrb \neq 0$.
Moreover, when
$\phi_1(S_1), \ldots, \phi_n(S_n)$ are representatives
of the equivalence classes $e_1, \ldots, e_n \in \calE$ 
whose corresponding variables $X_{e_i}$ appears in $q$,
the structure $\relb$ may be picked as
a structure of the form provided by Lemma~\ref{lemma:orthogonality}.
\end{theorem}

In order to establish this theorem,
we will make use of the following known fact 
concerning multivariate polynomials;
see for example~\cite[Lemma 2.1]{Alon99-nullstellensatz} for a proof.

\begin{prop} 
\label{prop:non-zero-multivariate-polynomial}
Let $p(x_1, \ldots, x_n)$ be a multivariate polynomial
in $n$ variables over a field $F$.
For each $i \in [n]$, let $d_i$ denote the degree of $p$
as a polynomial in $x_i$,
and suppose that $T_i \subseteq F$ is a set of size $d_i+1$ or greater.
Then, if $p$ is not the zero polynomial,
there exists a point $(t_1, \ldots, t_n) \in T_1 \times \cdots \times T_n$
such that $p(t_1, \ldots, t_n) \neq 0$.
\end{prop}

\begin{proof} (Theorem~\ref{thm:component-polynomials})
Let $\phi_1(S_1), \ldots, \phi_n(S_n)$ be as described in the theorem statement.
Let $m \geq 2$ be a value that exceeds the degree of 
each of the variables $X_{e_1}, \ldots, X_{e_n}$ in $q$, and 
apply Lemma~\ref{lemma:orthogonality}
to obtain structures 
$(\relb_{(a_1, \ldots, a_n)} ~|~ (a_1, \ldots, a_n) \in [m]^n)$
 and the corresponding functions
$f_1, \ldots, f_n: [m] \to \N$.
Evaluating $q$ on these structures amounts 
to evaluating $q$ when the variables
$(X_{e_1}, \ldots, X_{e_n})$ are given values in
$f_1([m]) \times \cdots \times f_n([m])$.
By Proposition~\ref{prop:non-zero-multivariate-polynomial},
$q$ must evaluate to a non-zero value on one of these structures.
\end{proof}


\begin{proof} (Theorem~\ref{thm:linearcombination})
Denote $\ell$ by $\sum_{i=1}^m c_i |\phi_i(V_i)|$
and let $(\rela_1, V_1), \ldots, (\rela_m, V_m)$ be the pairs
corresponding to the formulas $\phi_1(V_1), \ldots, \phi_m(V_m)$.
By rearranging the indices, we may assume for the sake of notation
that $\rela_1, \ldots, \rela_k$ are homomorphically equivalent
structures (where $k \in [m]$)
and that for no $i$ with $k < i \leq m$ does $\rela_i$ have
a homomorphism to $\rela_1$.

For any structure $\relb$, it holds that one of the values
$|\phi_1(\relb)|, \ldots, |\phi_k(\relb)|$
is non-zero if and only if all of them are.
Now, for each $i$, 
define $\widehat{\phi_i}(V_i)$ from $\phi_i(V_i)$
by removing non-liberal components,
that is, by removing each atom whose variables are all 
in a non-liberal component.
Note that for every $i \in [k]$ and for every structure $\relb$ 
 such that $|\phi_i(\relb)|>0$, we have that $|\widehat{\phi_i}(\relb)|=|\phi_i(\relb)|$. Since the $\phi_i$ are pairwise not counting equivalent, it follows that the $\widehat{\phi_i}$ are pairwise not counting equivalent.
For each formula $\widehat{\phi_i}(V_i)$, 
by considering its liberal connected components,
we may define $r_i$ to be a component polynomial
which is a product of variables from $\{ X_e ~|~ e \in \calE \}$
such that $|\widehat{\phi_i}(\relb)| = r_i\llb \relb \rrb$
for all finite structures $\relb$.
The products $r_1, \ldots, r_k$ are pairwise distinct,
so $r = c_1 r_1 + \cdots + c_k r_k$ is a non-zero component polynomial.
By applying Lemma~\ref{lemma:orthogonality} with $\rela = \rela_1$
and then invoking Theorem~\ref{thm:component-polynomials},
we obtain a finite structure $\reld$ of the form $\rela \times \relc$
such that $r \llb \reld \rrb \neq 0$. Since there is by assumption a homomorphism from $\rela_i$ to $\rela$ and there is by Theorem 5.4 a homomorphism from $\rela_i$ to $\relc$, we have $|\phi_i(\reld)| > 0$ for every $i\in [k]$. Consequently, $|\phi_i(\reld)| = |\widehat{\phi_i}(\relb)|$ by the observation from above.
Since no structure $\rela_i$ with $k < i \leq m$ maps homomorphically
to $\rela$, we have
$|\phi_{k+1}(\reld)| = \cdots = |\phi_{m}(\reld)| = 0$
and hence $\ell(\reld) = r \llb \reld \rrb \neq 0$.
\end{proof}

\section{Proofs of main theorems}
\label{sect:proofs-of-main-theorems}


The following lemma shows that if we are only interested in $\sh\PP$-representations of pp-formulas, we may restrict ourselves to basic $\sh\PP$-formulas.

\begin{lemma}\label{lem:basicisenoughsketch}
 Let $\phi$ be a pp-formula and let $\phi'$ be a $\sh\PP$-representation of $\phi$. Then there is a basic $\sh\PP$-sentence $\phi''$ that is also a $\sh\PP$-representation of $\phi$ such that $\width(\phi'')\le \width(\phi')$ and $\shwidth(\phi'')\le \shwidth(\phi')$.
\end{lemma}

\begin{proof} (sketch) 
Using the normalization result given by Lemma~\ref{lem:flattening}
and the results of Section~\ref{sect:quantifier-aware-width}, 
we can find  numbers $c'_i\in \mathbb{Z}$ and pp-formulas $\phi_i''$  such that for all structures $\relb$ 
\begin{align}\label{eq:linearcombsketch}[\relb, \phi'] = \sum_{i=1}^{\ell'} c_i' |\phi_i''(\relb)|,\end{align} where the $\phi_i''$ are pairwise not counting equivalent and have $\qaw(\phi_i'')\le \width(\phi')$. 
Note that (\ref{eq:linearcombsketch}) is a linear combination.

Now note that $[\relb, \phi']=|\phi(\relb)|$, which gives another linear combination. If follows with Theorem~\ref{thm:linearcombination} that 
(\ref{eq:linearcombsketch}) consists only of one summand with coefficient $1$. Now applying the results of Section~\ref{sect:quantifier-aware-width} yields the result.
\end{proof}

\licspods{
\begin{proof} (of Theorem~\ref{thm:shppequi} and Theorem~\ref{thm:shppcompute}; sketch) 
}{
\begin{proof}[of Theorem~\ref{thm:shppequi} and Theorem~\ref{thm:shppcompute}] 

\hspace{1pt}
\\
(sketch)
}
Observe that, by Lemma~\ref{lem:basicisenoughsketch}, we may  restrict attention to basic $\sh\PP$-sentence. Then both statements can be shown with the results of Section~\ref{sect:quantifier-aware-width} in a rather straightforward fashion.
\end{proof}



\licspods{
\begin{proof} (of Theorem~\ref{thm:shepequi}; sketch)
}{
\begin{proof}[of Theorem~\ref{thm:shepequi}](sketch)
}
 Let first $\Phi^+$ satisfy the tractability condition. Then there is a constant $k$ such that for all cores $\psi$ of pp-formulas in $\Phi^+$ we have $\tw(\psi)\le k$ and $\tw(\contract(\psi))+1 \le k$. 
 With the results of Section~\ref{sect:quantifier-aware-width}
 it follows that $\qaw(\psi)\le 2k$ for such cores $\psi$.
As discussed prior to the statement of Theorem~\ref{thm:shepequi},
we can express, for each $\phi\in \Phi$,
the function $|\phi(\cdot)|$ as a polynomial in
unknowns $|\psi(\cdot)|$ with $\psi \in \Phi^+$~\cite{ChenMengel16-pods-ep-counting}.
We can then substitute every term $|\psi'(\cdot)|$ by a $\sh\PP$-formula of width at most $2k$ with the results of Section~\ref{sect:quantifier-aware-width}. This yields a $\sh\PP$-formula of width at most $2k$.
 
For the other direction, assume the existence of a constant $k$ such that each ep-formula in $\Phi$ has a $\sh\EP$-representation $\phi'$ with 
$\width(\phi')\le k$. Let $\psi$ be a pp-formula from $\Phi^+$. We will show that $\psi$ is equivalent to a pp-formula of quantifier-aware width at most $k$ which completes the proof with Section~\ref{sect:quantifier-aware-width}. 

We first choose $\phi\in \Phi$ 
that witnesses $\psi \in \Phi^+$.
Let $\phi'$ be the $\sh\ep$-representation of $\phi$ of width at most $k$. We first construct a linear combination $\ell_1(.)$ as in (\ref{eq:linearcombsketch}). Note that all summands of $\ell_1(.)$ have quantfier-aware width at most $k$ as in the proof of Lemma~\ref{lem:basicisenoughsketch}.
 
We then construct a second representation of $|\phi(\relb)|$ as a linear combination $\ell_2(.)$. As in the first direction, for every structure we can express 
 $|\phi(\relb)|$ as a polynomial in unknowns of the form $|\theta(\relb)|$. Arithmetic simplifications and elimination of counting equivalent terms gives the second linear combination $\ell_2(.)$ computing $|\phi(\relb)|$. We then argue that a term of the form $c\cdot |\psi(\relb)|$ must appear in $\ell_2(.)$.
 
We have that $\ell_1$ and $\ell_2$ compute $|\phi(\relb)|$ and are hence equal.
 With Theorem~\ref{thm:linearcombination} it follows that $\ell_1$ contains a term that is counting equivalent to $\psi$. 
 Since all summands of $\ell_1(.)$ have quantifier-aware width at most $k$, the claim follows.
\end{proof}

\licspods{
	\begin{proof} (of Theorem~\ref{thm:shepcompute}; sketch)
}{
\begin{proof}[of Theorem~\ref{thm:shepcompute}] (sketch)
}
 Using the results of Section~\ref{sect:quantifier-aware-width} 
 and the normalization result of Lemma~\ref{lem:flattening},
  $\phi$ can be turned into a linear combination as in (\ref{eq:linearcombsketch}). Moreover, starting with any such representation yields the same linear combination up to counting equivalence of the summands by Theorem~\ref{thm:linearcombination}. Now turning this linear combination into a $\sh\PP$-formula yields a $\sh\PP$-representation and minimizing the width of the summands with Theorem~\ref{thm:shppcompute} gives a representation of optimal width.
\end{proof}

\bibliographystyle{IEEEtran}
\bibliography{hubiebib}

\newpage

\section{Material}

\pods{
\subsection{Remark on quantifier-aware width}

\begin{remark}\label{rem:star}
The quantifier-aware width can be arbitrarily higher than their treewidth. To see this consider the formula $\phi = \exists z \bigwedge_{i\in [n]} E(x_i, z)$. The primal graph of $\phi$ is a star, so it has treewidth $1$. We claim that the quantifier-aware width of $\phi$ is $n+1$. To see this, observe first that the free variables $x_1, \ldots, x_n$ must appear above $\topp(z)$ in a bag of any quantifier-aware tree decomposition. But since $x_i z$ is an edge in the primal graph for every $i$, the variable $x_i$ must also appear in a common bag with $z$ and consequently also in $\topp(z)$. Thus $\topp(z)$ must contain $n+1$ variables. 
\end{remark}
}


\subsection{Proof of Lemma~\ref{lem:qawcompare}}

\begin{proof}[of Lemma~\ref{lem:qawcompare}]
The main idea of the proof is that forcing the free variables of an 
$\exists$-component to appear above the quantified variables has a very similar effect as connecting them to a clique in the construction of the $\contract(\phi)$. So from a quantifier-aware decomposition we get a decomposition of $\contract(\phi)$ by restricting to the tree decomposition to the free variables. For the other direction, we can add the quantified variables to a tree decomposition of $\contract(\phi)$ in a straightforward way. 
We now give the details.

For the first inequality, observe first that any quantifier-aware tree decomposition of $\phi$ is a tree decomposition of $\phi$, so $\tw(\phi) + 1\le \qaw(\phi)$ is obvious. 
To prove the inequality $\tw(\contract(\phi))+1\le \qaw(\phi)$,
 let $(T, (B(t))_{t\in T})$ be a quantifier-aware tree decomposition of $\phi$. Introduce for every $\exists$-component of $\phi$ a new vertex $v_C$. Then substitute in every $B(t)$ every non-free variable $x$ of $\phi$ by $v_C$ where $C$ is the $\exists$-component that contains $x$. Call the result $B'(t)$. We claim that $(T, (B'(t))_{t\in T})$ is a tree decomposition of $\contract(\phi)$. To see this, note that by the same argument as in Remark~\ref{rem:star} we have for every $\exists$-component $C$ a bag that contains 
 $(V(C) \cap \free(\phi)) \cup \{ v_C \}$.

 For the second inequality, first compute a tree decomposition 
 $(T, (B(t))_{t\in T})$ of the contraction $\contract(\phi)$. Note that for every $\exists$-component $C$ of $\phi$ there is a bag $B(t^*)$ that contains $V(C)\cap \free(\phi)$, because this variable set forms a clique in $\contract(\phi)$ and it is well known that for every clique in a graph every tree decomposition must contain a bag that contains this clique completely. Now compute a tree decomposition $(T', (B'(t))_{t\in T'})$ for $G[V(C)\setminus \free(\phi)]$, where $G$ is the primal graph of $\phi$. Then construct  for every $t\in T'$ a new bag $B''(t) := B'(t) \cup (V(C) \cap \free(\phi))$. Finally, connect $T$ to $T'$ by connecting an arbitrary vertex of $T'$ to $t^*$. Doing this for every $\exists$-component yields a quantifier-aware tree decomposition of $\phi$. Moreover, the width of the decomposition is at most $\tw(\contract(\phi))+\tw(\phi)+1$ which completes the proof.
\end{proof}

\subsection{Proof of Lemma~\ref{lem:computeqaw}}
\begin{proof}
The idea of the proof is to add edges to the graph of $\phi$ 
in such a way that the treewidth of the resulting graph is exactly $\qaw(\phi)$; then, we apply standard algorithms for computing treewidth.
 To this end, let $\phi$ be a pp-formula with primal graph $G$ and $S:= \free(\phi)$. For each $\exists$-component $C$ of $\phi$ choose a vertex $x_C\in V(C)\setminus \free(\phi)$ and 
 connect it to all vertices $y\in V(C)\cap \free(\phi)$;
 moreover, connect the vertices in $V(C)\cap \free(\phi)$ by a clique. 
 Call the resulting graph $G'$. We will show that the minimum of $\tw(G')+1$ over the choices of the $x_C$ is $\qaw(\phi)$.
 
 We first show that for every choice of the $x_C$ we have $\tw(G')+1\ge \qaw(\phi)$. To see this, fix a tree decomposition 
 $(T, (B(t))_{t\in T})$ of $G'$. Since $V(C)\setminus \free(\phi)$ is connected in $G'$, the bags containing $V(C)\setminus \free(\phi)$ are contained in a subtree $T'$ of $T$. Moreover, because $\{x_c\}\cup (V(C)\cap \free(\phi))$ is a clique in $G'$, we know that $\{x_c\}\cup (V(C)\cap \free(\phi))\subseteq B(t^*)$ for some $t^*$ in $T'$. Since none of the vertices in $V(C)\setminus \free(\phi)$ have any neighbors outside of $V(C)$, we may assume that $t^*$ is the root of $T'$. Then it is easy to see that $(T, (B_t)_{t\in T})$ can be turned into a quantifier-aware tree decomposition: We only have to potentially add a new bag $B(t^{**}):= V(C)\cap \free(\phi)$ and a vertex $t^{**}$ in the decomposition. Then connect $t^{**}$ to $t^*$ and its parent and delete the edge between $t^*$ and its parent.
 
 For the other direction, let $(T, (B(t))_{t\in T})$ be a quantifier-aware tree decomposition of $\phi$. We will show that it is also a tree decomposition of $G'$ for a choice of the $x_C$. First note that by the same argument as before, the vertices of $V(C)\cap \free(\phi)$ are contained in a subtree $T'$ of $T$. Let $x_C'$ be the only variable of $V(C)\setminus \free(\phi)$ that is contained in $B(r)$ where $B(r)$ is the root of $T'$. Note that by the same argument as in Remark~\ref{rem:star}, we know that $B(r)$ contains $V(C)\cap \free(\phi)$. Thus $B(r)$ covers all edges introduced in the construction of $G'$ when choosing $x_C=x_C'$. Thus $(T, (B(t))_{t\in T})$ is indeed a tree decomposition of $G'$ for the right choice of the $x_C$.
 
 Since computing tree decompositions is fixed parameter tractable 
 when parameterized by the treewidth (see e.g.~\cite{FlumGrohe06-parameterizedcomplexity}), the only problem left to solve is the right choice of the $x_C$. But since the quantified variables of the different $\exists$-components are independent, we can do this choice independently for every $\exists$-component $C$ as follows: Construct $G''$ by choosing a vertex $x_C\in V(C)\cap \free(\phi)$ and proceed as in the construction of $G'$. Now for all other $\exists$-components $C'$ connect $V(C')\cap \free(\phi)$ to a clique and delete all variables in $V(C')\setminus \free(\phi)$. Clearly, trying all potential choices of $x_C$ lets us optimize the choice for $C$. Doing this for all $\exists$-components gives the desired choice and thus the optimal quantifier-aware tree decomposition.
\end{proof}

\subsection{Proof of Lemma~\ref{lem:shppandpp}}

\begin{proof} 
The proof relies on the observation that the condition on $\topp(x)$ in quantifier-aware tree decompositions corresponds closely to the fact that free variables can only be closed by a $P$-quantifier in $\sh\PP$-formulas after the contained pp-formula has been casted by a $C$-quantifier. With this in mind,  a $\sh\PP$-formula is transformed into a pp-formula
by making use of and inducting on
the $\sh\PP$-formula's structure (viewed as a tree). 
The other direction is similar.

a) 
Let $\phi'$ be a basic $\sh\PP$-sentence. By potentially renaming
quantified variables, make sure that every variable in $\phi'$ is either
free or quantified exactly once. We construct $\phi$ by deleting all
$C$-, $E$- and $P$-quantifiers and substituting all $\times$ by $\land$.
Obviously, the result is a pp-formula. Note that for every subformula
$\psi'$ of $\phi'$, we have an associated subformula $\psi$ of $\phi$.
For every subformula $\psi$ of $\phi$ we define $\lib(\psi)$ to be the
variables of $\psi'$ that are not quantified in $\psi$. We claim that
for all $\sh$-subformulas $\psi'$ and every assignment $h$ to
$\free(\psi')$,
 \begin{equation}\label{eq:shppandpp2} [\relb,
\psi'](h) = |\{ h': \lib(\psi) \rightarrow B \mid h' \text{ extends } h,
(\relb, h')\models \psi\}|.\end{equation}

We show (\ref{eq:shppandpp2}) by induction on the structure of basic
$\sh\PP$-formulas. If $\psi' = C \psi''$ for a pp-formula $\psi''$, then
we actually have $\psi = \psi''$. Moreover, $h$ assigns to values to all
liberal variables of $\psi$, so both sides of (\ref{eq:shppandpp2}) are
$1$ if and only if $h$ satisfies $\psi$. If $\psi'= EV\psi''$ or $\psi'=
PV \psi''$, we get (\ref{eq:shppandpp2}) directly from the semantics of
$\sh$-formulas and induction. Finally, if $\psi'=\psi_1' \times
\psi_2'$, we have that $\lib(\psi_1)\cap \lib(\psi_2) \subseteq
\lib(\psi)$.
Therefore, $\lib(\psi_1)$ and $\lib(\psi_2)$ overlap 
only at variables where a mapping $h': \lib(\psi) \to B$ 
is defined, and
(\ref{eq:shppandpp2}) follows easily.

 
 
 It remains to show the inequalities of the width measures. To this end, consider the syntax tree $T$ of $\phi'$. For every node $t$ of $T$, define $B(t) := \free(\phi_t')$ where $\phi_t'$ is the subformula of $\phi'$ that has $t$ as its root. Note that $(T, (B(t))_{t\in T})$ satisfies the connectivity condition and is thus a tree decomposition of $\phi$ of width $\width(\phi')-1$. Also, $(T, (B(t))_{t\in T})$ is quantifier-aware because in $\phi'$ existential quantification is only allowed in the pp-part in which all free variables of $\phi$ are still free. This shows $\width(\phi')\ge \qaw(\phi)$. Now observe that by deleting all bags that contain quantified variables we end up with a tree decomposition for $\contract(\phi)$. This shows $\shwidth(\phi) \ge \tw(\contract(\phi))+1$.
 
 b) Let now $\phi$ be a pp-formula and let $(T, (B(t))_{t\in T})$ be a nice quantifier-aware tree decomposition of $\phi$ of width $k-1$. For every $\exists$-component $C$ of $\phi$, the vertices $V(C)\setminus \free(\phi)$ all lie in the bags of a subtree $T_C$ of $T$. Moreover, we may w.l.o.g.~assume that the bags in $T_C$ do not contain any vertices not in $V(C)$. Finally, we have that the bag $B(r_C)$ where $r_c$ is the root of $T_C$ contains $V(C)\cap \free(\phi)$, because $(T, (B(t))_{t\in T})$ is quantifier-aware. The results of \cite{DalmauKolaitisVardi02-treewidth} assure that there is a pp-formula $\phi_C$ of width $k$ that is logically equivalent to the pp-formula that we get by restricting $\phi$ to the atoms that have all of their variables in $V(C)$.
 
 We now construct for every $t\in T$ such that $B(t)$ does not contain any quantified variables of $\phi$ a basic $\sh\PP$-sentence $\phi_t'$. So let $t$ be a node of $t$ with the desired properties. Let $\atom(t)$ be the atoms of $\phi$ containing only variables in $B(t)$ and set $\bar{\phi}_{t}:= \prod_{\psi \in \atom(t)} C \psi$. If $t$ has no children, set $\phi_t := \bar{\phi}_t$. 
 
 If $t$ has a child $t'$ such that $B(t')\setminus \free(\phi)\ne \emptyset$, then $t$ has only that one child because $(T, (B(t))_{t\in T})$ is nice. Let $C$ be the unique $\exists$-component of the variable in $B(t')\setminus \free(\phi)$. We set $\phi_t:= \bar{\phi}_t \times C \phi_C$.
 
 If $t$ has a child $t'$ such that $B(t')\setminus \free(\phi)= \emptyset$ and a variable $x$ is forgotten when going from $t'$ to $t$, then set $\phi_t := P x \phi_{t'}$.
 
 If $t$ has a child $t'$ and a variable $x$ is introduced when going from $t'$ to $t$, then set $\phi_t := \bar{\phi}_t \times E x \phi_{t'}$.
 
 If $t$ has two children $t_1$ and $t_2$, then note that 
 $$B(t_1)\setminus \free(\phi)= B(t_2)\setminus \free(\phi)= \emptyset.$$ 
 Moreover, $\free(\phi_{t_1}) = \free(\phi_{t_2})$. We define $\phi_t := \phi_{t_1} \times \phi_{t_2}$.
 
 Set $\phi':= P \free(\phi_r) \phi_r$ where $r$ is the root of $\phi$.
 
 An easy induction along the construction of $\phi$ similar to that in a) shows that $\phi'$ does indeed compute the correct value for every structure $\relb$. Moreover, the width of $\phi'$ is at most $k$ which completes the proof.
 
 If we do not have a bound on $\qaw(\phi)$ but only on $\tw(\contract(\phi))$, the same construction as above yields the the bound $\tw(\contract(\phi)) + 1 \ge \shwidth(\phi')$. The only difference is that we do not have to bound the width of the pp-formulas with \cite{DalmauKolaitisVardi02-treewidth}.
\end{proof}

\section{Proof of Lemma~\ref{lem:turnPP}}

\begin{proof}
In a first step, we transform $\psi$ into a logically equivalent 
ep-formula
$\psi^d = \bigvee_{i=1}^s \psi_i$ where the $\psi_i$ are pp-formulas;
this can be done without increasing width~\cite{Chen14-existentialpositive}.
 Then we claim that $\phi$ is logically equivalent to 
 $$\phi'=\sum_{J\subseteq [s], J\ne \emptyset} (EL (-1)^{|J|+1}) \prod_{i\in J} C(\psi_i, L).$$ First note that this is a well-formed $\sh$-formula, because for all additions and multiplications the free variables of all operands are $L$. It remains to show that $\phi'$ is logically equivalent to $\phi$. So fix $\relb$ and $h:L\rightarrow B$. 
 
 If $h$ does not satisfy $\psi$, then $[\relb, \phi](h)=0$. Since $h$ does not satisfy any $\psi_i$, it is easy to see that $[\relb, \phi](h)=0$ as well.
 
Assume $h$ satisfies $\psi$; say that $\psi_1, \ldots, \psi_\ell$
are the disjuncts that it satisfies.
By definition $[\relb, \phi](h)=1$.
Also,
$[\relb, \phi'](h)$
is equal to
$$\sum_{J\subseteq [s], J\ne \emptyset} (-1)^{|J|+1} \prod_{i\in J} [\relb,C(\psi_i, L)](h)$$ 
which in turn is equal to $\sum_{J\subseteq [\ell], J\ne \emptyset} (-1)^{|J|+1}  $;
this latter quantity is equal to the sum 
$$\ell - {\ell \choose 2} + {\ell \choose 3} - \ldots \pm {\ell \choose \ell}  = (1-1)^\ell + 1 = 1.$$
 It is readily seen that the width of $\phi'$ is not bigger than that of $\phi$.
\end{proof}

\subsection{Proofs for Section~\ref{sect:independence-theorem}}


\begin{proof}[Explanation for Lemma~\ref{lemma:commutativity}]
Let $p$ be a univariate polynomial with positive integer coefficients 
and variable $X$.
Fix a representation of $p$ as a term with $1$ and $X$ as the inputs
and where addition and multiplication are the operations.
For a vocabulary $\tau$, let $\reli_{\tau}$ 
 denote 
the $\tau$-structure 
with universe $\{ a \}$ and where
each relation symbol $R \in \tau$ has 
$R^{\reli_{\tau}} =\{ (a,\ldots, a)\}$. 
For each structure $\relb$ over vocabulary $\tau$,
we define $p(\relb)$ as the $\tau$-structure obtained by
evaluating the representation of $p$ by interpreting $1$
as $\reli_\tau$, $X$ as $\relb$,
addition as the disjoint union $\uplus$ of two structures,
and multiplication as the product of two structures.

Here, by the disjoint union $\rela \uplus \relb$ of two structures
$\rela$, $\relb$, we mean the structure obtained as follows:
rename the elements of $\rela$, $\relb$ so that their universes
$A, B$ are disjoint, and then take the structure 
with universe $A \cup B$
 where $R^{\rela \uplus \relb} = R^{\rela} \cup R^{\relb}$
 for each relation symbol $R$.
 By the product $\rela \times \relb$ 
 of two structures $\rela, \relb$, we mean the structure
 with universe $A \times B$ and 
 where $R^{\rela \times \relb} = \{ ((a_1, b_1), \ldots, (a_k, b_k)) ~|~ (a_1, \ldots, a_k) \in R^{\rela}, (b_1, \ldots, b_k) \in R^{\relb} \}$.

Lemma~\ref{lemma:commutativity}
can be proved by a straightforward induction on the structure of $p$,
using the observations that 
$|\phi(\reli_\tau)| = 1$,
$|\phi(\reld \uplus \reld')| = |\phi(\reld)| + |\phi(\reld')|$,
and
$|\phi(\reld \times \reld')| = |\phi(\reld)|\cdot|\phi(\reld')|$.
\end{proof}


\subsection{Full proofs for Section~\ref{sect:proofs-of-main-theorems}}

Remember that we call a $\sh\PP$-formula $\phi$ basic
if it does not contain $+$ nor subformulas of the form $n$, 
where $n \in \Z$. The following lemma shows that if we are only interested in $\sh\PP$-representations of pp-formulas, we may restrict ourselves to basic $\sh\PP$-formulas.

\begin{lemma}\label{lem:basicisenough}
 Let $\phi$ be a pp-formula and let $\phi'$ be a $\sh\PP$-representation of $\phi$. Then there is a basic $\sh\PP$-sentence $\phi''$ that is also a $\sh\PP$-representation of $\phi$ such that $\width(\phi'')\le \width(\phi')$ and $\shwidth(\phi'')\le \shwidth(\phi')$.
\end{lemma}

\begin{proof}
 With Lemma~\ref{lem:flattening} we may assume that $\phi'$ is flat, so let $\phi' = \sum_{i=1}^\ell \psi_i'\times \phi_i'$ 
 where each $\psi_i'$ is constant and each $\phi_i'$ is basic. 
As the $\sh\PP$-formulas $\phi'_i$ are basic,
by Lemma~\ref{lem:shppandpp} (a)
we obtain that there are pp-formulas $\phi_i$ representing them,
where for each finite structure $\relb$
we have $|\phi_i(\relb)| = [\relb, \phi_i']$.
Set $c_i$ to be the value that the constant formula
$\psi_i'$ evaluates to.
We have that, for every structure $\relb$, 
$[\relb, \phi'] = \sum_{i=1}^\ell c_i |\phi_i(\relb)|$,
where $c_i\in \mathbb{Z}$ and $\phi_i$ is a pp-formula.



Now combine the summands of counting equivalent pp-formulas to get a linear combination with 
 \begin{align}\label{eq:linearcomb}[\relb, \phi'] = \sum_{i=1}^{\ell'} c_i' |\phi_i''(\relb)|,\end{align} 
 where the $\phi_i''$ are pairwise not counting equivalent pp-formulas and $c_i'\in \mathbb{Z}\setminus \{0\}$. Note that for all $\phi_i''$ we have by Lemma~\ref{lem:shppandpp} that 
 $\qaw(\phi_i'')\le \width(\phi')$
 and
 $\tw(\contract(\phi_i''))+1 \le \shwidth(\phi')$.
 
 Since $\phi'$ is a $\sh\PP$-representation of $\phi$, we have 
 $[\relb, \phi']=|\phi(\relb)|$
 for all structures $\relb$. 
So, we have two linear combinations, $|\phi(\relb)|$
and that of (\ref{eq:linearcomb}) that evaluate to the same value.
By taking the difference and invoking 
Theorem~\ref{thm:linearcombination}, we obtain
that these two linear combinations are equal, up to 
counting equivalence of the pp-formulas.
Since the linear combination $|\phi(\relb)|$ only has one summand
with coefficient $1$,
it follows that the linear combination in (\ref{eq:linearcomb}) consists only of one summand with coefficient $1$. Let $|\psi(\relb)|$ be that summand. We have 
that $\psi$ and $\phi$ are counting equivalent.
   Now we apply Lemma~\ref{lem:shppandpp} (b) on input $\psi$ to construct a formula with the desired properties.
\end{proof}

Another ingredient that we will use in the proofs of the main theorems is a syntactic characterization of counting equivalence. To this end, we say that
two pp-formulas
$(\rela, V)$, $(\rela', V')$ 
 over the same signature
 are \emph{renaming equivalent}
if there exist surjections 
$h_1: V \to V'$ and $h_2: V' \to V$
 that can be extended to homomorphisms 
 $\bar{h} : \rela \rightarrow \rela'$ and 
 $\bar{h'}:\rela'\rightarrow \rela$, respectively.
 
 A crucial result of~\cite{ChenMengel16-pods-ep-counting}
is that renaming equivalence is a syntactic characterization of counting equivalence.
 
 \begin{theorem}
 \label{thm:renamingandcountingold}
\cite[Theorem 5.4]{ChenMengel16-pods-ep-counting}
Any two pp-formulas $\phi_1(V), \phi_2(V')$ are counting equivalent
if and only if they are renaming equivalent.
\end{theorem}

We will make use of the following easy corollary.

\begin{corollary}\label{cor:renamingandcounting}
 Let $\phi(V)$ and $\phi'(V')$ be counting equivalent pp-formulas. Then there is a pp-formula $\psi$ that is logically equivalent to $\phi$ such that $\qaw(\psi)\le \qaw(\phi')$.
\end{corollary}
\begin{proof}
By Theorem~\ref{thm:renamingandcountingold}, we have that $\phi$ and $\phi'$ are renaming equivalent. 
Construct $\psi(V)$ from $\phi'$ by renaming variables
so that $\phi(V)$ and $\psi(V)$ are renaming equivalent
via the identity mapping $i: V \to V$ on $V$.
We claim that $\psi$ and $\phi$ are logically equivalent, that is,
for any structure $\relb$
and any mapping $f: V \to B$,
it holds that $\relb,f \models \psi$ iff $\relb,f \models \phi$.

We prove the forward direction; the backward direction is analogous.
Let $[\rela_\psi, V]$ and $[\rela_\phi, V]$ be the structure
views of the pp-formulas $\psi$ and $\phi$, respectively.
Suppose that $\relb,f \models \psi$.
By the work of Chandra and Merlin~\cite{ChandraMerlin77-optimal} discussed in the preliminaries,
we obtain that there is an extension $g: A_\psi \to B$
of $f$ that is a homomorphism from $\rela_\psi$ to $\relb$.
Since $\phi$ and $\psi$ are renaming equivalent via the identity
mapping $i$ on $V$,
there exists a homomorphism $h$ from $\rela_\phi$ to $\rela_\psi$
that extends $i$.
By composing $h$ and $g$,
we obtain an extension of $f$ that is a homomorphism
from $\rela_\phi$ to $\relb$.
By the work of Chandra and Merlin~\cite{ChandraMerlin77-optimal} discussed in the preliminaries,
we obtain that $\relb,f \models \phi$.
\end{proof}

\begin{proof}[of Theorem~\ref{thm:shppequi}]
 We start with the second statement. Let first $\Phi$ satisfy the contraction condition. Then there is a constant $k$ such that for all cores $\phi$ of pp-formulas in $\Phi$ we have 
 $\tw(\contract(\phi))+1 \le k$. But then Lemma~\ref{lem:shppandpp} yields basic $\sh\PP$-representations $\phi'$ with $\shwidth(\phi')\le k$.
 

 Now assume there is a constant $k$ such that every formula $\phi$ in $\Phi$ has a $\sh\PP$-representa\-tion $\phi'$ such that $\shwidth(\phi')\le k$. By Lemma~\ref{lem:basicisenough} we may assume that $\phi'$ is basic. Then by Lemma~\ref{lem:shppandpp} we obtain a pp-formula $\phi''$ that is counting equivalent to $\phi$ and such that we have 
 $\tw(\contract(\phi'')) +1 \le k$. 
 As $\phi$ and $\phi''$ are counting equivalent, by 
 Theorem~\ref{thm:renamingandcountingold},
 they are renaming equivalent.
 Consequently, the treewidth of $\contract(\cdot)$
 applied to the cores of $\phi$ and $\phi''$, is the same.
 Since passing to the core does not increase the treewidth of
 $\contract(\cdot)$, the treewidth of
 $\contract(\cdot)$ applied to the core of $\phi$ is $\leq k-1$.
 By Corollary~\ref{cor:renamingandcounting}, 
 we have that $\phi$ is logically equivalent to a formula $\psi$ 
 with $\tw(\contract(\psi))+1 \le k$. Consequently, $\Phi$ satisfies the contraction condition.
 
 For the first statement, let first $\Phi$ satisfy the tractability condition. Then there is a constant $k$ such that for all cores $\phi$ of pp-formulas in $\Phi$ we have $\tw(\phi)\le k$ and $\tw(\contract(\phi))+1\le k$. It follows that $\qaw(\phi)\le 2k$ by Lemma~\ref{lem:qawcompare}. Then Lemma~\ref{lem:shppandpp} yields a basic $\sh\PP$-representation 
 $\phi'$ with $\width(\phi')\le 2k$.
 
 Now assume there is a constant $k$ such that every formula $\phi$ in $\Phi$ has a $\sh\PP$-representa\-tion $\phi'$ such that $\width(\phi')\le k$. We may again assume that $\phi'$ is basic. Then Lemma~\ref{lem:shppandpp} gives a pp-formula $\phi''$ that is counting equivalent to $\phi$ such that $\qaw(\phi) \le k$. Using Corollary~\ref{cor:renamingandcounting} it follows that $\phi$ is logically equivalent to a formula $\psi$ with $\qaw(\phi)\le k$. Now applying Lemma~\ref{lem:qawcompare} shows that $\Phi$ satisfies the tractability condition.
\end{proof}

\begin{proof}[of Theorem~\ref{thm:shppcompute}]
 By Lemma~\ref{lem:basicisenough} we may assume that the desired $\sh\PP$-representa\-tion is basic. It then suffices to minimize the width of the factor that is not constant. To do so, we translate to a pp-formula with Lemma~\ref{lem:shppandpp} (a), then minimize the width with Lemma~\ref{lem:computeqaw} and finally translate back to a basic $\sh\PP$-formula with Lemma~\ref{lem:shppandpp} (b).
\end{proof}

In what follows,
we will use the following lemma that summarizes the main construction of the proof of Theorem~3.1 in~\cite{ChenMengel16-pods-ep-counting}.

\begin{lemma}\label{lem:classificationdetails}
 For every set $\Phi$ of ep-formulas there is a set $\Phi^+$ of pp-formulas such that the following holds: 
  For every $\phi\in \Phi$ there are pp-formulas $\phi_1, \ldots, \phi_\ell\in \Phi^+$, pp-sentences $\psi_1, \ldots, \psi_t\in \Phi^+$ and non-zero integers $c_1, \ldots, c_l$ such that the following holds for every structure $\relb$:
 \begin{itemize}
 \item If $\relb$ satisfies one of the sentences $\psi_i$, then 
 $|\phi(\relb)| = |B|^{|\free(\phi)|}$.
 \item Otherwise, we have $|\phi(\relb)| = \sum_{i=1}^\ell c_i |\phi_i(\relb)|$.
 \end{itemize}
 Moreover, the $\phi_i$ are pairwise not counting equivalent and there is no homomorphism from any sentence $\psi_i$ to any $\psi_j$ for $i\ne j$ and no homomorphism to any $\phi_j$.
 
 Finally, for every $\phi'\in \Phi^+$ we have that $\phi'$ appears as a $\phi_i$ or $\psi_i$ as above for a $\phi\in \Phi$.
\end{lemma}

\begin{proof}[of Theorem~\ref{thm:shepequi}]
First suppose that $\Phi^+$ satisfies the tractability condition. Then there is a constant $k$ such that for all cores $\psi$ of pp-formulas in $\Phi^+$ we have $\tw(\psi)\le k$ and $\tw(\contract(\psi))+1 \le k$. 
 With Lemma~\ref{lem:qawcompare} it follows that $\qaw(\psi)\le 2k$ for the cores $\psi$ of the pp-formulas in $\Phi$. 
 Then Lemma~\ref{lem:classificationdetails} provides
pp-formulas $\phi_1, \ldots, \phi_{\ell} \in \Phi^+$,
coefficients $c_1, \ldots, c_\ell$,
and
pp-sentences
$\psi_1, \ldots, \psi_t \in \Phi^+$ 
such that the following holds for every structure $\relb$: If $\relb$ satisfies any sentence $\psi_i$, then $|\phi(\relb)|= |B|^{|\free(\phi)|}$; otherwise, $|\phi(\relb)| = \sum_{i=1}^\ell c_i |\phi_i(\relb)|$.
Let $|\free(\phi)|=r$. 
 Then 
 $|\phi(\relb)|$ is equal to 
$(\prod_{i=1}^t (1-|\psi_i(\relb)|)) (\sum_{i=1}^\ell c_i |\phi_i(\relb)|) + |B|^r (1-\prod_{i=1}^t (1-|\psi_i(\relb)|))$.
 Now let $\phi_1', \ldots, \phi_\ell'$ be the $\sh$-formulas we get by applying Lemma~\ref{lem:shppandpp} on $\phi_1, \ldots, \phi_\ell$
 and then multiplying respectively by the $\sh$-formulas
 \[E \free(\phi) c_1, \ldots, E \free(\phi) c_\ell\]
 that correspond to the coefficients $c_i$.
 By Lemma~\ref{lem:shppandpp}, we have that the width of each $\psi'_i$ is
 bounded above by $2k$.

 Then set 
$\psi$ to be $P\lib(\phi)$
applied to the sum of \[(E\free(\phi)(1-C(\bigvee_{i=1}^t \psi_i, \emptyset))) \times (\sum_{i=1}^\ell \phi'_i)\] and
\[E\free(\phi) C(\bigvee_{i=1}^t \psi_i, \emptyset).\]
 Clearly, $\psi$ is a $\sh\EP$-representation of $\phi$. Moreover, the width of $\psi$ is the maximal width of the $\psi_i$ and 
 the $\phi'_i$, which is bounded above by $2k$.

 The first direction of the second item can be proved completely analogously using $\tw(\contract(.))$ and $\sh\width$ of $\qaw$ and $\width$.

For the other direction, assume that there is a constant $k$ such that every ep-formula in $\Phi$ possesses a $\sh\EP$-representation $\phi'$ with $\width(\phi)\le k$. Let $\psi$ be a pp-formula from $\Phi^+$. We will show that $\psi$ is equivalent to a pp-formula of quantifier-aware width at most $k$ which completes the proof with Lemma~\ref{lem:shppandpp}. 

We first choose $\phi\in \Phi$ such that $\psi$ appears as a $\phi_i$ or $\psi_i$ as in Lemma~\ref{lem:classificationdetails}. Let $\phi'$ be the $\sh\ep$-representation of $\phi$ of width at most $k$. With Lemma~\ref{lem:flattening} we may assume that $\phi'$ is flat, i.e., it has the form $\sum_{i=1}^\ell \theta_i'\times \phi_i',$ where the $\theta_i'$ are constant and the $\phi_i'$ are basic. As in the proof of Lemma~\ref{lem:basicisenough}, this yields for every $\relb$
 \begin{align}\label{eq:linearcomb2}|\phi(\relb)| = \sum_{i=1}^{\bar{\ell}} \bar{c}_i |\bar{\phi}_i(\relb)|,\end{align} where the $\phi_i$ are pairwise not counting equivalent. Moreover, $\qaw(\bar{\phi}_i)\le k$.
 
We will now construct a second representation of $|\phi(\relb)|$ as a linear combination. To this end, let $|\lib(\phi)|=r$ and let $\phi_1, \ldots, \phi_\ell$, $\psi_1, \ldots, \psi_t$ and $c_1, \ldots, c_\ell$ be as in Lemma~\ref{lem:classificationdetails}. Then as before, for every structure 
 $|\phi(\relb)| = \prod_{i=1}^t(1- |\psi_i(\relb)|) (\sum_{i=1}^\ell c_i |\phi_i(\relb)|) + |B|^r (1-\prod_{i=1}^t (1-|\psi_i(\relb)|)).$ 
Now multiplying the righthand side out, we get that $|\phi(\relb)|$ can be expressed as a weighted sum of  terms of the form $|\phi_i(\relb)| \cdot \prod_{j\in J} |\psi_j(\relb)|$ and of the form $|B|^r \prod_{j\in J} |\psi_j(\relb)|$. 
These terms are equivalent to \[|(\phi_i \land \bigwedge_{j\in J}\psi_j)(\relb)|\] and \[|(\psi^r\land \bigwedge_{j \in J}\psi_j)(\relb)|,\]
since the $\psi_j$ are sentences; here, $\psi^r$ is the true formula with liberal variables $\lib(\phi)$.

Now combine counting equivalent summands to get a linear combination \begin{align}\label{eq:linearcomb3}|\phi(\relb)| = \sum_{i=1}^\ell c_i' |\phi_i'(\relb)|.\end{align} 

We claim that the linear combination in (\ref{eq:linearcomb3}) contains $c \cdot |\psi(\relb)|$ or $c \cdot |\psi^r \land \psi(\relb)|$for some $c\ne 0$ as a summand. To see this, consider first the case that $\psi$ is not a sentence. Note that in this case $c\cdot \psi(\relb)$ appears in the weighted sum from Lemma~\ref{lem:classificationdetails} and thus (for $J= \emptyset$) also in the sum we get before combining counting equivalent summands. Moreover, $\psi$ is not counting equivalent to any other summand $\psi'$ in this sum. To see this, note first that by Lemma~\ref{lem:classificationdetails} there is no homomorphism from any of the $\psi_j$ to $\psi$. Moreover, the $\phi_i$ are pairwise not counting equivalent. Thus $c|\psi|$ appears in (\ref{eq:linearcomb3}) in this case. If $\psi$ is one of the sentences $\psi_i$, then $\psi^r \land \psi$ is contained in the sum before combining counting equivalent summands (for $J=\{i\}$). Obviously, $\psi^r\land \psi$ is not counting equivalent to any conjunction that contains a non-sentences $\psi_i$. Moreover, since there are no homomorphism between $\psi=\psi_i$ and $\psi_j$ by Lemma~\ref{lem:classificationdetails}, $\psi^r\land \psi$ is not counting equivalent to any other summand, so it must be contained in (\ref{eq:linearcomb3}).

It follows that (\ref{eq:linearcomb2}) and (\ref{eq:linearcomb3}) give two linear combinations that are equal for every structure $\relb$. Using Theorem~\ref{thm:linearcombination} it follows that (\ref{eq:linearcomb2}) contains a summand $c\cdot |\bar{\psi}|$ that is counting equivalent to $\psi$. Moreover, $\qaw(\bar{\psi})\le k$. 
Invoking Corollary~\ref{cor:renamingandcounting},
it follows 
that $\psi$ is logically equivalent to a formula with quantifier-aware width at most $k$.

Again, the second direction of the second item follows completely analogously.
\end{proof}

\begin{proof}[of Theorem~\ref{thm:shepcompute}] 
The algorithm first comptues a
$\sh\ep$-representation $\psi$
of the input $\phi \in \Phi$,
such as that given by Proposition~\ref{prop:simple-representation}.
 By Lemma~\ref{lem:flattening} we may assume that $\psi$ is flat, i.e., it has the form 
 $\psi=\sum_{i=1}^\ell \psi_i \times \phi_i$ where the $\psi_i$ are constant and the $\phi_i$ are basic. Note that $\width(\psi)=\max_{i\in \ell}(\width(\phi_i))$.
 
 Note that every flat $\sh\EP$-representation of $\phi$ can be turned into a linear combination as in (\ref{eq:linearcomb}). Moreover, starting with any such representation yields the same linear combination up to counting equivalence of the summands by Theorem~\ref{thm:linearcombination}. Now turning this linear combination into a flat $\sh\PP$-formula,
 one obtains a $\sh\PP$-representation; minimizing the width of the summands with Theorem~\ref{thm:shppcompute} gives a representation of optimal width.
\end{proof}

\end{document}